\documentclass[sn-mathphys]{sn-jnl}

\usepackage{graphicx}%
\usepackage{amssymb,amsfonts}%
\usepackage{amsthm}%
\usepackage{textcomp}%
\usepackage{manyfoot}%
\usepackage{listings}%
\usepackage{subcaption} 
\usepackage{float}
\usepackage{dblfloatfix}
\usepackage{thmtools}
\usepackage{algpseudocode}
\usepackage{algorithm}
\usepackage{cancel} 
\usepackage{enumerate}
\usepackage{soul} 
\usepackage{mathtools} 

\usepackage{tikz} 
\usetikzlibrary{tikzmark}
\usetikzlibrary{automata,positioning,
    shadows,
    shapes.multipart,shapes.misc,backgrounds,shapes.geometric}

\tikzstyle{activity}=[rectangle split ,rounded corners,rectangle split parts=2, fill=orange, draw ,text centered]

\definecolor{light-gray}{gray}{0.85}
\definecolor{lighter-gray}{gray}{0.95} 
\definecolor{turquoise}{rgb}{0.0, 0.71, 0.81}
\definecolor{peach}{rgb}{0.97, 0.59, 0.35}
\definecolor{YellowGreen}{rgb}{0.60, 0.80, 0.44}
\definecolor{MyLavender}{RGB}{164, 105, 209}

\tikzset{
    activity/.style = {
            rectangle split,
            rectangle split parts=2,
            rectangle split part fill={light-gray,lighter-gray}, 
            draw, rounded corners, inner sep=1.5,outer sep=0,
            align=left, text=black }
}
\tikzset{
    action rectangle/.style = {
            rectangle, font=\footnotesize,
            align=center, text=black }
}
\tikzset{every state/.style={minimum size=1pt } }
\tikzset{
    my rectangle/.style={
            blue,
            rounded corners
        }
}
\tikzset{elliptic state/.style={draw,ellipse}}
\newcommand*{\QEDB}{\null\nobreak \hfill\scalebox{0.9}{\ensuremath{\blacksquare}}}%

\newcommand{\eventset}{{\mathbb{E}}}
\newcommand{\outcomeset}{{\mathbb{O}}}
\newcommand{\resourseset}{{\mathbb{S}}}
\newcommand{\timefunc}{{\mathcal{T}}}
\newcommand{\mapfunc}{{\mathcal{M}}}
\newcommand{\ioa}{{\mathit{Y}}}
\newcommand{\acts}{{\mathit{Acts}}}
\newcommand{\actsuni}{{\mathbb{ACT}}}
\newcommand{\Dee}{{\mathbb{D}}}
\newcommand{\DA}{{\mathbb{D}_\mathbb{A}}} 
\newcommand{\DE}{{\mathbb{D}_\mathbb{E}}}
\newcommand{\DLC}{{\mathbb{D}_\mathit{LC}}} 
\newcommand{\DAC}{{\mathbb{D}_\mathit{AC}}} 
\newcommand{\DaC}{{\mathbb{D}_\mathit{aC}}} 
\newcommand{\behfunc}{{\mathcal{B}}}
\newcommand{\Vee}{{\mathcal{V}}}
\newcommand{\emptysequence}{{\tau}} 
\newtheorem{theorem}{Theorem}
\newtheorem{proposition}{Proposition} 
\newtheorem{lemma}{Lemma}
\newtheorem{corollary}{Corollary}
\newtheorem{definition}{Definition}%
\declaretheorem[style=definition]{assumption}

\begin{document}

\title[Time- and Behavior-Preserving Execution of Determinate Supervisory Control]{Time- and Behavior-Preserving Execution of Determinate Supervisory Control}

\author*[1]{\fnm{Alireza} \sur{Mohamadkhani}}\email{a.mohammadkhani@tue.nl}

\author[1]{\fnm{Marc} \sur{Geilen}}\email{m.c.w.geilen@tue.nl}

\author[1]{\fnm{Jeroen} \sur{Voeten}}\email{j.p.m.voeten@tue.nl}
\author[1,2]{\fnm{Twan} \sur{Basten}}\email{a.a.basten@tue.nl}

\affil[1]{\orgdiv{Department of Electrical Engineering}, \orgname{Eindhoven University of Technology}, \orgaddress{\street{De Zaale}, \city{Eindhoven},  \country{The Netherlands}}}
 
\affil[2]{\orgdiv{ESI (TNO)},    \city{Eindhoven},   \country{The Netherlands}}

\abstract{
    The \emph{activity framework} is a promising model-based design approach for Flexible Manufacturing Systems (FMS).
    It is used in industry for specification and analysis of FMS.
    It provides an intuitive specification language with a hierarchical view of the system's actions and events, activities built from them, and
    an automaton that captures the overall behavior of the system in terms of sequences of activities corresponding to its accepted words.
    It also provides a scalable timing analysis method using max-plus linear systems theory.
    The framework currently requires manual implementation of the supervisory controller that governs the system behavior.
    This is labor-intensive and error-prone.
    In this article, we turn the framework into a model-driven approach by introducing an execution architecture and
    execution engine that allow a specification to be executed in a time- and behavior-preserving fashion.
    We prove that the architecture and engine preserve the specified ordering of actions and events in
    the specification as well as the timing thereof up to a specified bound.
    We validate our approach on a prototype production system.
}




\maketitle
 
\section{Introduction}\label{section:introduction}
  
A flexible manufacturing system (FMS) is an integrated, computer-controlled, complex of
automated material handling devices and machine tools that can
simultaneously process multiple types of parts while reducing the overhead time, effort, and cost of changing production line
configurations in response to changing requirements \cite{stecke1983formulation}.
FMS try to achieve the efficiency of automated mass production, while having the
flexibility to respond to changes in market demand and produce several part types simultaneously.

The flexibility of FMS comes with increased system complexity, making it more difficult to
design and implement these systems \cite{elmaraghy2005flexible}. FMS typically include different subsystems that are required to
work together to deliver the required results. This makes the design and implementation of these systems
an increasingly labor-intensive challenge that benefits from design automation.

Model-based design \cite{nicolescu2018model} of FMS is a design approach in which we
model and analyze the FMS to obtain a design that satisfies the requirements of the system.
It helps the design process of these systems by providing models
that are both understandable by engineers and amenable to mathematical analysis.
Model-driven design \cite{schmidt2006model} is a specialization of model-based design where
models drive the implementation.
This has the potential to reduce costs and time to deploy these systems. It also
significantly reduces the odds of human error in the implementation process and yields a system
that preserves the prescribed model behavior and satisfies performance requirements.
It can further give engineering teams a concise view of how the system behaves and allows automated
optimizations that are hard to perform manually.

The activity framework \cite{sanden2016} is a modular and scalable model-based design methodology for FMS.
It focuses on the \emph{supervisory control} layer of the FMS that controls the system by orchestrating
the actions and events that the \emph{plant} layer executes.
The framework is being used by some of the world leaders in the semiconductor industry (such as ASML, VDL-ETG, and ITEC)
for specification, mechanical layout analysis, and performance analysis of FMS \nolinebreak\cite{sanden2015, lsat}.
It provides an intuitive specification language that abstracts system actions and events and their dependencies into
\emph{activities}. Each activity represents a determinate
piece of system behavior. As an example, consider the activity of turning a part during production with a
gripper, consisting of gripping, movement, and rotation actions; initial movement precedes gripping,
which in turn precedes further movement with rotation occurring in parallel. Determinacy in the context of the activity framework
means that variations that respect action dependencies within an activity do not influence the behavior
of the system, implying that the system is robust against timing variations.
The activity abstraction provides a scalable timing analysis
technique rooted in max-plus linear systems theory. The framework uses a modular approach to logistics
controller specification using automata.
It also provides best-case and worst-case makespan and throughput analysis
through max-plus computations.

The activity framework currently requires manual implementation of the supervisory
control layer by the engineering teams after its design is finalized.
Manual implementation is error-prone and cumbersome. It can introduce unspecified behavior to the system due to
lack of clear mutual understanding of the engineering teams. It can also fail to achieve the performance targets
that the model has predicted.

\begin{figure}[t]
    \centering
    \scalebox{1.0}{
        \includegraphics[clip,width=0.6\columnwidth]{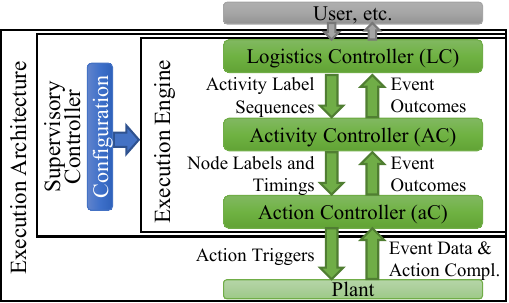}
    }
    \caption{The execution architecture}
    \label{fig:architecture}
\end{figure}

In this article, we enable model-driven design of FMS based on the activity framework by introducing
an execution architecture and an execution engine, called the Activity Execution Engine.
Using the concepts introduced in \cite{mohamadkhani}, we define the activity framework with \emph{event feedback},
which is a mechanism to provide feedback from the plant layer of an FMS to the supervisory control layer.
To specify the logistics controller, we introduce deterministic I/O automata to specify the behavior of the system in the form of
sequences of activities that are executed in response to specified event feedback the outcome of which is resolved at execution time.


The execution architecture that we present in this article,
illustrated in Figure \ref{fig:architecture}, consists of our proposed execution engine, which is the
implementation of the supervisory controller, and the physical plant.
The execution engine is configured with an activity model of an FMS with event feedback and a deterministic I/O automaton
as input, and executes it on the plant. It is organized in three sublayers,
responsible for logistics, activity execution, and action execution, respectively, with well-defined interfaces.
We formulate a correctness relation between the behavior of the model and the physical execution
that captures the preservation of the ordering of actions and events of a behavior and their timing up to a specified bound, as well
as respecting the outcomes of events received during the execution.
We prove that our execution engine is time- and behavior-preserving according to that relation.
The actions and events that are executed only depend on the outcome of specified events and
not on any timing deviation of actions or events in the execution. Thus, the behavior of the system is determinate
up to the resolution of the outcomes of events.
Our execution engine thus guarantees determinate behavior of the system despite timing variations that may happen in execution.
We validate our solution on a prototype production system.

The remainder of this article is organized as follows.
Section \nolinebreak\ref{section:motivexample}
presents a motivating example. In Section \ref{section:model}, we
describe the modeling framework that we use in this article to specify system behavior and timing, and present a model of a
simple running example.
In Section \ref{section:relation}, we discuss the timing and behavior relation that our execution architecture
preserves between the model and the execution. Section \ref{section:executionArch}
describes the plant model that is assumed in the architecture, as well as the execution engine, and proves that
the execution engine preserves timing and behavior of the specification.
Then, we validate our proposed execution architecture and execution engine in Section \ref{section:validation}.
In Section \ref{section:relatedWrok}, we examine related literature and compare the state of the art to our approach.
Finally, Section \ref{section:conclusion} concludes our work.

\section{Motivating Example}\label{section:motivexample}

In this section, we describe a motivating example to illustrate our approach.
The eXplore Cyber-Physical Systems (xCPS) \cite{xcps} platform (Figure \ref{fig:xcps}) is a small-scale
production line. It is capable of producing different products. One of its production scenarios
is that it assembles \emph{top} pieces and \emph{bottom} pieces, and
outputs the assembled products. The machine captures the features and behavior of FMS present in
industry. The pieces go through several processing steps in
the machine using belts and stoppers: first, a \emph{gantry arm} collects them and puts them onto an input belt, which
moves them to a \emph{turner} station where all tops are correctly oriented. The bottoms are assumed to be correctly oriented so they pass
through this station.  
After the turner station, the bottoms are pushed to an \emph{indexing table} by a \emph{switch}, and the tops move to an \emph{assembler}.
The proper routing of pieces is managed through a number of sensors, stoppers, and switches.
When a correctly-oriented top arrives at the assembler, the assembler picks it up and waits for
the bottom that was pushed to the indexing table. The table rotates and positions the bottom under the assembler station. The assembler then
puts the top that it picked on the bottom on the indexing table. Afterward the assembled product is pushed to the output belt by a \emph{switch} and
is removed from the assembly line.

\begin{figure}[t!]
    \centering
    \begin{subfigure}{0.8\linewidth}
        \centering
        \scalebox{1}{
            \includegraphics[clip,width=0.7\columnwidth]{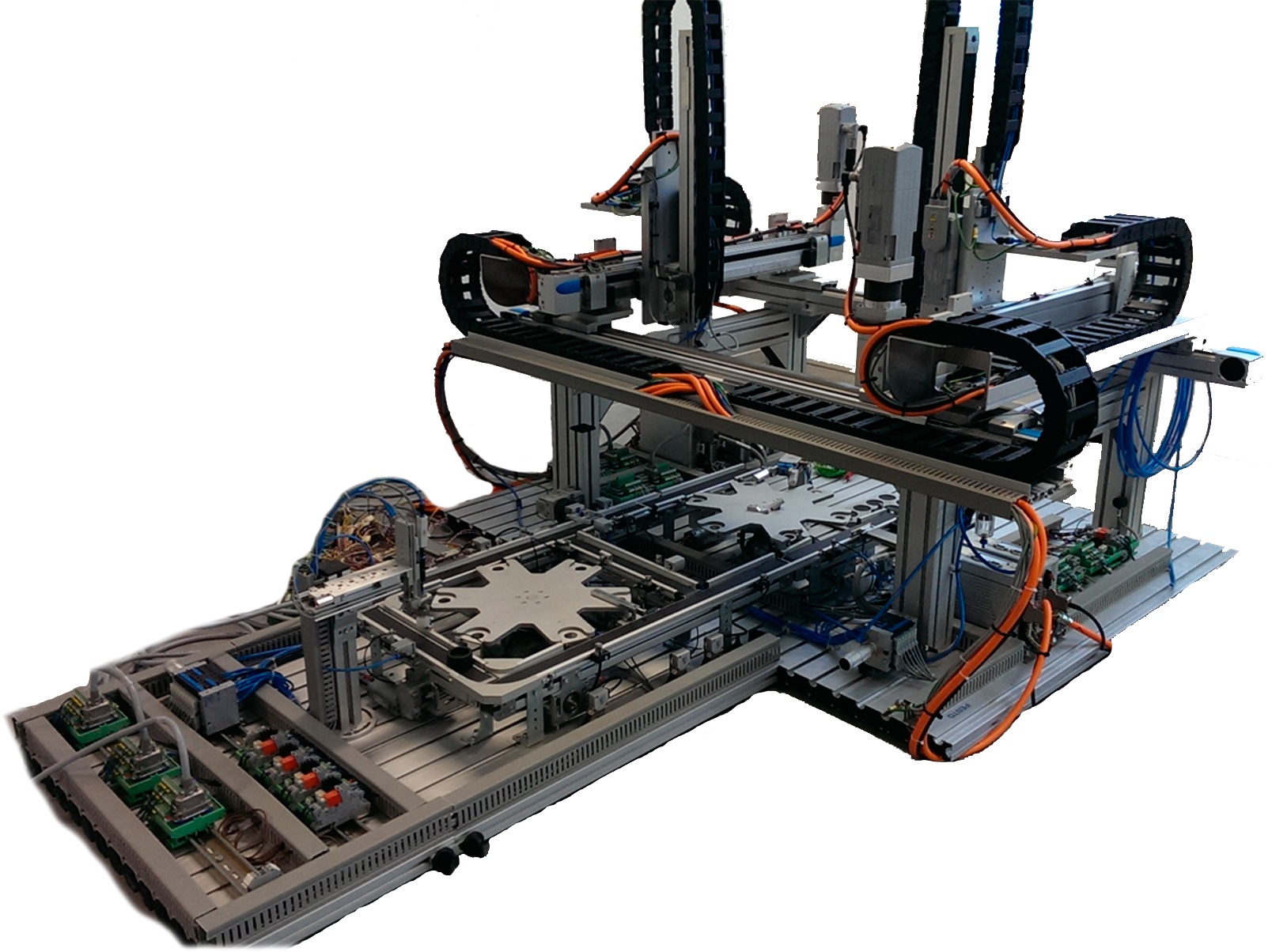}
        }
        \caption{The physical layout}
    \end{subfigure}
    \begin{subfigure}{1\linewidth}
        \centering
        \scalebox{1}{
            \includegraphics[clip,width=0.7\columnwidth]{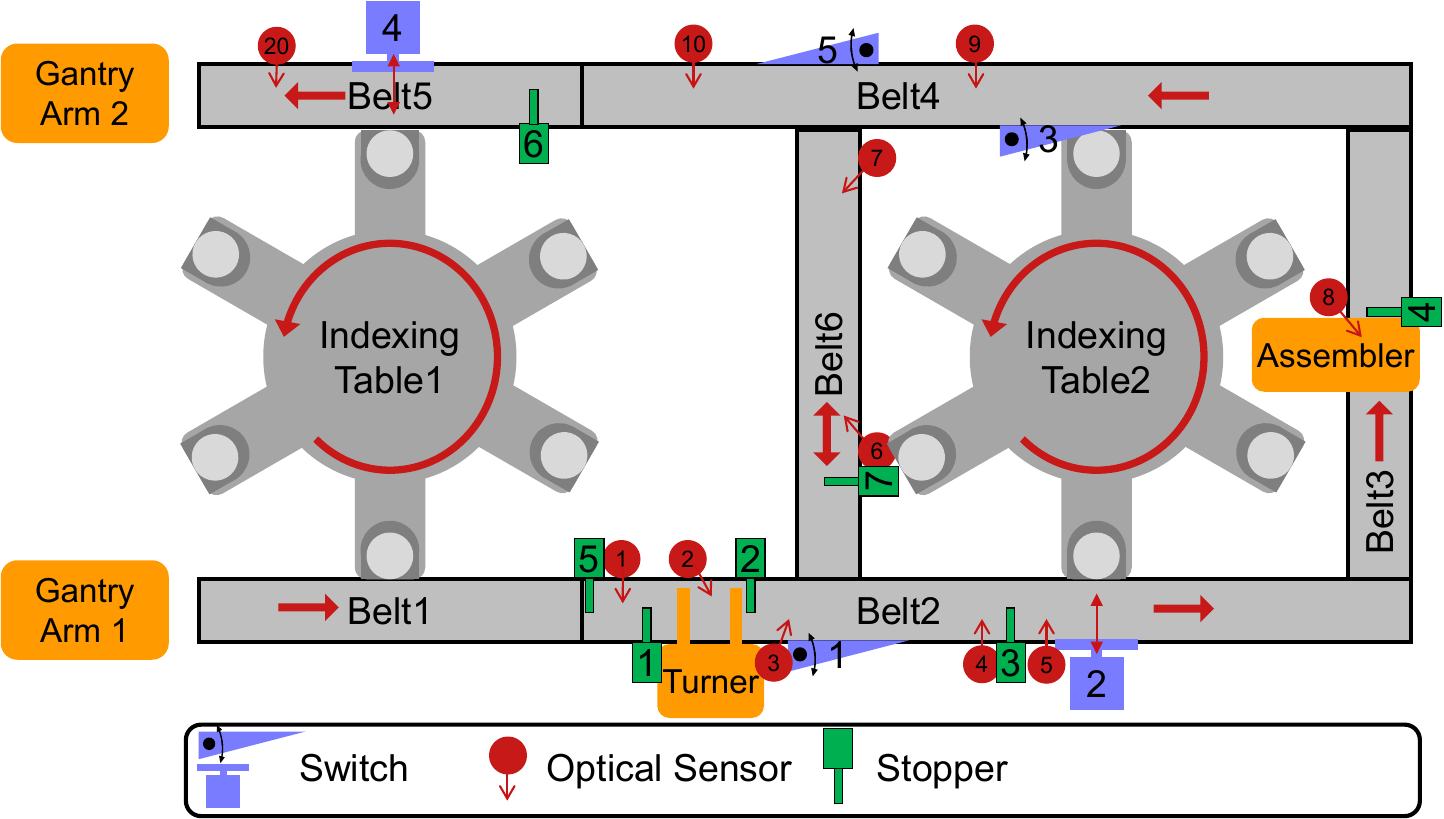}
        }
        \caption{Schematic view}
        \label{fig:xcps-schematic}
    \end{subfigure}
    \caption{The xCPS }
    \label{fig:xcps}
\end{figure}

We model the xCPS using the activity framework.
The basis of an activity model is a set of \emph{resources}, for instance the turner station or the assembler station.
Resources can have \emph{peripherals} that are able to perform \emph{actions}. For example, the
turner station has a motor that can perform
the actions to pick up a piece, another motor that can rotate a piece,
and a gripper that can grab or release a piece.
The supervisory control of a machine also requires feedback from the machine. For example, if an assembler fails to pick up a top piece,
the controller needs to respond accordingly. For this purpose, the activity framework
incorporates event feedback. Event feedback is introduced in \cite{mohamadkhani} in the form of the switched max-plus linear system characterization.
We define the corresponding operational characterization in this article.
\emph{Events} are planned notifications of outcomes of activities
from the plant to the supervisory controller that enable the
controller to react to occurrences of interest within the plant.
Events have a set of possible \emph{outcomes} that are non-deterministic and are resolved at runtime.
The xCPS, for instance, has an event notifying success or failure, emitted by the assembler every time it attempts to pick up a top piece.
To assemble a product, the \emph{assemble} activity is performed where the assembler on belt 3
goes down, grabs a top piece, moves up, extends to the indexing table, goes down, and finally
releases the piece. A sensor indicates whether the assembler was able to pick up the piece or not.
This constitutes the event that the activity assemble emits with two possible outcomes, \emph{success} and \emph{failure}.
In the case of a failure, the piece is moved to a reject bin by the belts.

To implement an FMS, the engineers face challenges such as
the integral control of \emph{synchronization},  \emph{concurrency},  \emph{timing} and  \emph{pipelining} of
actions and activities while properly and timely responding to event feedback from the plant.
Handling all these aspects simultaneously is difficult using current engineering methods.
In the case of the xCPS,
we have different stations around the indexing table that have to work concurrently and in a synchronized way
to ensure proper assembly.
The transport of a piece by belts and the indexing table alignment with the belts have to happen concurrently.
The pushing on and off of the table have to happen simultaneously, while the assembly has to wait for
correct alignment. So we have to maintain timing, concurrency, and the correct ordering of actions that we execute.
The activities around the table have to be pipelined if we aim for maximum system utilization and productivity.
Another challenge is that certain actions are time-sensitive and
executing them at the wrong time leads to incorrect behavior.
In xCPS, pushing the bottoms onto the table is one such action.
There is no sensor in front of switch~2 to detect pieces, so we have to rely on timing.
We release the stopper that is holding a bottom piece and
activate switch 2 at very precisely defined times.

The example above shows the delicate details of the xCPS that have to be accounted for
to ensure proper behavior of the system. These challenges are present in many FMS.
Current state of practice leaves the implementation of the supervisory control layer to manual programming, which is error-prone.
We propose a model-driven approach, in which an activity model of the system can be used to automatically configure standard,
reusable execution software to realize its supervisory control including concurrency, pipelining, precisely timed actions and event feedback.

\section{System Specification, Behavior and Timing}\label{section:model}

In this section, we explain the activity framework \cite{sanden2016},
and its extension with event-feedback, as introduced in \nolinebreak\cite{mohamadkhani}, for specification and analysis purposes.
The main purpose is to precisely define system specification in the framework,
and how it captures system behavior and timing, as it is preserved by our execution engine.
Therefore, this section is organized into three subsections, corresponding to specification, behavior, and timing.


\subsection{System Specification}

A system model consists of a set of activities,
an I/O automaton that specifies the logistics controller that captures the possible sequences of activities that can be executed,
a set of events, a set of event outcomes, and the relation between them.
Activities consist of actions and events and the dependencies between them.
The I/O automaton captures the response of the logistics controller to the outcomes of events
in terms of the activity sequences that it executes.

The activity framework models activities as Directed Acyclic Graphs (DAGs),
where nodes are either \emph{actions} performed on \emph{peripherals}, \emph{claims} or \emph{releases}
of a \emph{resource}, or \emph{events} being emitted by the activity.
The concept of resource is used as a model abstraction that allows imposing constraints
on the model to ensure correctness of behavior and validity of execution.
Resources have a (possibly empty) set of peripherals. Resources can be \emph{claimed}
and \emph{released}.

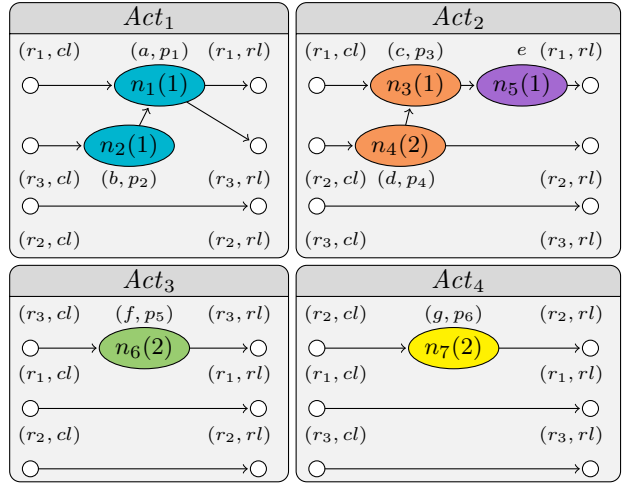
\begin{figure}[t]
    \centering
    \begin{tikzpicture}
        \tikzstyle{every state}=[  inner sep=1pt, minimum size=0pt]
        \node (n1) [activity]
        {
            \textbf{$\mathit{Act}_1$}
            \nodepart[align=left]{two}

            \begin{tikzpicture}[shorten >=1pt,on grid,auto]

                \node[state] (q_0) [fill=white, minimum size=6pt] {};
                \node[action rectangle] (l_0) [ above right=0.45cm and 0.25cm of q_0]{$(r_1,cl)$};
                \node[state] (q_1) [right=1.7 of q_0, fill=turquoise, style={font=\small}, style={draw,ellipse}] {$n_1 (1)$};
                \node[action rectangle] (l_0) [ above=0.45 of q_1]{$(a,p_1)$};
                \node[state] (q_3) [right=1.3 of q_1, fill=white, minimum size=6pt] {};
                \node[action rectangle] (l_0) [ above left=0.45 and 0.25 of q_3]{$(r_1,rl)$};

                \node[state] (q_4) [below=0.8 of q_0,fill=white, minimum size=6pt] {};
                \node[action rectangle] (l_0) [ below right=0.45cm and 0.25cm of q_4]{$(r_3,cl)$};
                \node[state] (q_5) [right=1.3 of q_4, fill=turquoise, style={font=\small}, style={draw,ellipse}] {$n_2 (1)$};
                \node[action rectangle] (l_0) [ below=0.45 of q_5]{$(b,p_2)$};
                \node[state] (q_7) [below=0.8 of q_3, fill=white, minimum size=6pt] {};
                \node[action rectangle] (l_0) [  below left=0.45 and 0.25 of q_7]{$(r_3,rl)$};

                \node[state] (q_8) [below=0.8 of q_4,fill=white, minimum size=6pt] {};
                \node[action rectangle] (l_0) [ below right=0.45cm and 0.25cm of q_8]{$(r_2,cl)$};
                \node[state] (q_9) [below=0.8 of q_7, fill=white, minimum size=6pt] {};
                \node[action rectangle] (l_0) [  below left=0.45 and 0.25 of q_9]{$(r_2,rl)$};

                \path[->]
                (q_0) edge   (q_1)
                (q_1) edge   (q_3)
                (q_4) edge   (q_5)
                (q_1) edge   (q_7)
                (q_5) edge  (q_1)
                (q_8) edge  (q_9)
                ;
            \end{tikzpicture}
        };

        \node (n2) [activity, right=0.1 of n1 ]
        {
            \textbf{$\mathit{Act}_2$}
            \nodepart[align=left]{two}

            \begin{tikzpicture}[shorten >=1pt,on grid,auto]

                \node[state] (q_0) [fill=white, minimum size=6pt] {};
                \node[action rectangle] (l_0) [ above right=0.45cm and 0.25cm  of q_0]{$(r_1,cl)$};
                \node[state] (q_1) [right=1.3 of q_0, fill=peach , style={font=\small}, style={draw,ellipse}] {$n_3(1)$};
                \node[action rectangle] (l_0) [ above=0.45 of q_1]{$(c,p_3)$};
                \node[state] (q_2) [right=1.4 of q_1, fill=MyLavender, style={font=\small},style={draw,ellipse} ] {$n_5(1)$};
                \node[action rectangle] (l_0) [ above=0.45 of q_2]{$e$};
                \node[state] (q_3) [right=0.9 of q_2, fill=white, minimum size=6pt] {};
                \node[action rectangle] (l_0) [ above left=0.45 and 0.25 of q_3]{$(r_1,rl)$};

                \node[state] (q_4) [below=0.8 of q_0,fill=white, minimum size=6pt] {};
                \node[action rectangle] (l_0) [ below right=0.45cm and 0.25cm  of q_4]{$(r_2,cl)$};
                \node[state] (q_5) [right=1.1 of q_4, fill=peach, style={font=\small}, style={draw,ellipse} ] {$n_4(2)$};
                \node[action rectangle] (l_0) [ below=0.45 of q_5]{$\;\;   (d,p_4)$};
                \node[state] (q_7) [below=0.8  of q_3, fill=white, minimum size=6pt] { };
                \node[action rectangle] (l_0) [ below left=0.45 and 0.25 of q_7]{$(r_2,rl)$};

                \node[state] (q_8) [below=0.8 of q_4,fill=white, minimum size=6pt] {};
                \node[action rectangle] (l_0) [ below right=0.45cm and 0.25cm  of q_8]{$(r_3,cl)$};
                \node[state] (q_9) [below=0.8  of q_7, fill=white, minimum size=6pt] { };
                \node[action rectangle] (l_0) [ below left=0.45 and 0.25 of q_9]{$(r_3,rl)$};

                \path[->]
                (q_0) edge   (q_1)
                (q_2) edge   (q_3)
                (q_4) edge   (q_5)
                (q_1) edge   (q_2)
                (q_5) edge   (q_1)
                (q_5) edge   (q_7)
                (q_8) edge   (q_9)
                ;
            \end{tikzpicture}
        };
        \node (n3) [activity, below=0.1 of n1 ]
        {
            \textbf{$\mathit{Act}_3$}
            \nodepart[align=left]{two}

            \begin{tikzpicture}[shorten >=1pt,on grid,auto]

                \node[state] (q_0) [fill=white, minimum size=6pt] {};
                \node[action rectangle] (l_0) [ above right=0.45cm and 0.25cm  of q_0]{$(r_3,cl)$};
                \node[state] (q_1) [right=1.5 of q_0, fill=YellowGreen, style={font=\small}, style={draw,ellipse}] {$n_6(2)$};
                \node[action rectangle] (l_0) [ above=0.45 of q_1]{$(f,p_5)$};
                \node[state] (q_3) [right=1.5 of q_1, fill=white, minimum size=6pt] {};
                \node[action rectangle] (l_0) [ above left=0.45 and 0.25 of q_3]{$(r_3,rl)$};

                \node[state] (q_8) [below=0.8 of q_0,fill=white, minimum size=6pt] {};
                \node[action rectangle] (l_0) [ above right=0.45cm and 0.25cm  of q_8]{$(r_1,cl)$};
                \node[state] (q_9) [below=0.8  of q_3, fill=white, minimum size=6pt] { };
                \node[action rectangle] (l_0) [ above left=0.45 and 0.25 of q_9]{$(r_1,rl)$};

                \node[state] (q_80) [below=0.8 of q_8,fill=white, minimum size=6pt] {};
                \node[action rectangle] (l_0) [ above right=0.45cm and 0.25cm  of q_80]{$(r_2,cl)$};
                \node[state] (q_90) [below=0.8  of q_9, fill=white, minimum size=6pt] { };
                \node[action rectangle] (l_0) [ above left=0.45 and 0.25 of q_90]{$(r_2,rl)$};

                \path[->]
                (q_0) edge   (q_1)
                (q_1) edge   (q_3)
                (q_8) edge   (q_9)
                (q_80) edge   (q_90)
                ;
            \end{tikzpicture}
        };
        \node (n4) [activity, right=0.1 of n3  ]
        {
            \textbf{$\mathit{Act}_4$}
            \nodepart[align=left]{two}

            \begin{tikzpicture}[shorten >=1pt,on grid,auto]

                \node[state] (q_4) [ fill=white, minimum size=6pt] { };
                \node[action rectangle] (l_0) [ above right=0.45cm and 0.25cm  of q_4]{$(r_2,cl)$};
                \node[state] (q_5) [right=1.8 of q_4, fill=yellow, style={font=\small}, style={draw,ellipse}] {$n_7(2)$};
                \node[action rectangle] (l_0) [ above=0.45 of q_5]{$(g,p_6)$};
                \node[state] (q_7) [right=1.8 of q_5, fill=white, minimum size=6pt] { };
                \node[action rectangle] (l_0) [ above left=0.45 and 0.25 of q_7]{$(r_2,rl)$};

                \node[state] (q_8) [below=0.8 of q_4,fill=white, minimum size=6pt] {};
                \node[action rectangle] (l_0) [ above right=0.45cm and 0.25cm  of q_8]{$(r_1,cl)$};
                \node[state] (q_9) [below=0.8  of q_7, fill=white, minimum size=6pt] { };
                \node[action rectangle] (l_0) [ above left=0.45 and 0.25 of q_9]{$(r_1,rl)$};

                \node[state] (q_80) [below=0.8 of q_8,fill=white, minimum size=6pt] {};
                \node[action rectangle] (l_0) [ above right=0.45cm and 0.25cm  of q_80]{$(r_3,cl)$};
                \node[state] (q_90) [below=0.8  of q_9, fill=white, minimum size=6pt] { };
                \node[action rectangle] (l_0) [ above left=0.45 and 0.25 of q_90]{$(r_3,rl)$};

                \path[->]
                (q_4) edge   (q_5)
                (q_5) edge   (q_7)
                (q_8) edge   (q_9)
                (q_80) edge   (q_90)
                ;
            \end{tikzpicture}
        };
    \end{tikzpicture}

    \caption{Activities $\mathit{Act}_1,\mathit{Act}_2,\mathit{Act}_3,$  and $\mathit{Act}_4$ of the running example. Colors are used to distinguish between the action nodes of different activities and also between action and event nodes.}
    \label{fig:activities-1234}
\end{figure}

Figure \ref{fig:activities-1234} depicts activities $\mathit{Act}_1,\mathit{Act}_2,\mathit{Act}_3,$
and $\mathit{Act}_4$ of our running example.
Action nodes are colored and labeled (for example $n_1$) and their timing durations are shown between parentheses.
We omit node labels for resource nodes for brevity.
We use labels on top or at the bottom of each node to depict whether they are actions on peripherals,
events, or claims or releases of resources. For example, node $n_1$ represents action $a$ being performed on peripheral $p_1$.
Arrows denote dependencies. For example, $n_2$ has to complete before $n_1$ can start.
Event nodes denote a notification from the plant layer to the supervisory control layer.
Each event has one or more known possible outcomes and depending on the outcome that is received, the logistics controller
responds by executing a different sequence of activities.

The activity framework provides a hierarchical model with
different levels of abstraction, from individual actions, to activities,
sequences of activities, and an automaton that captures sequences of activities.
The automaton uses activities on its transition labels, making it
significantly smaller than a logistics controller for individual actions,
which enhances scalability.

The activity framework formalizes these concepts as follows.
Let $\mathbb{A}$ denote the set of actions, let $\resourseset$ denote the set of resources,
and let $\mathbb{P}$ denote the set of peripherals.
We assume a function $\mathcal{R} : \mathbb{P} \rightarrow \resourseset$, where $\mathcal{R}(p)$  is the
resource that contains $p$. To model event feedback, we further let $\eventset$ be the set
of events and $\outcomeset$ be the set of event outcomes.
We define an activity as follows.

\begin{definition} \label{def:activity}
    (Activity). An activity is a quadruple ($\mathit{Nod},\rightarrow\nolinebreak, \mapfunc, \timefunc$) where
    \begin{itemize}
        \item $(\mathit{Nod},{\rightarrow})$ is a DAG, with $\mathit{Nod}$ a set of nodes and ${\rightarrow} \subseteq \mathit{Nod} \times \mathit{Nod} $ a set of dependencies;
        \item $\mapfunc : \mathit{Nod} \rightarrow \mathbb{A} \times \mathbb{P} \cup \resourseset \times \{rl, cl\} \cup \eventset$ is a function that maps every node to an action label $(a,p)$ denoting action $a$ on peripheral $p$, a release or claim label where  $(r,\mathit{rl})$ denotes a release  of resource $r$ and $(r,\mathit{cl})$ denotes a claim of resource $r$, or an event;
        \item $\timefunc: \mathit{Nod} \rightarrow \mathbb{R}_{\geq 0}$ is a function that assigns a non-negative time value to every node such that $\timefunc(n)=0$ if $\mapfunc(n)\in \resourseset \times \{rl, cl\}$; $\timefunc$ denotes \emph{specified duration} if $\mapfunc(n)\in \mathbb{A} \times \mathbb{P}$, and it denotes \emph{specified delay} if $\mapfunc(n)\in \eventset$.
    \end{itemize}
    We write $n_1 \rightarrow n_2$ to denote that $(n_1, n_2) \in {\rightarrow}$. We denote the transitive closure of $\rightarrow$ by $\rightarrow^+$. We define the predecessor nodes of $n$ as:
    \begin{equation*}
        \mathit{Pred}(n) = \{ n_{in} \in \mathit{Nod} \mid n_{in} \rightarrow n\}
    \end{equation*}
    An activity satisfies the following constraints:
    \begin{enumerate}[I.]
        \item All nodes mapped to the same peripheral are sequentially ordered to avoid self-concurrency on a single peripheral:\\
              $ \forall_{n_1 , n_2 \in \mathit{Nod}, a_1, a_2 \in \mathbb{A}, p \in \mathbb{P}} \: n_1 \neq n_2 \wedge \mapfunc(n_1) = (a_1,p) \wedge \mapfunc(n_2) = (a_2,p)    \Rightarrow   n_1 \rightarrow^+ n_2 \vee n_2 \rightarrow^+ n_1$
        \item Each resource is claimed exactly once: \label{bullet:cl-all}\\
              $  \forall_{r \in \resourseset} \: \exists_{n_1 \in \mathit{Nod}} \: \mapfunc(n_1) = (r,cl) \wedge    \nexists_{n_2 \in \mathit{Nod}} \: n_1\neq{}n_2 \wedge \mapfunc(n_2) = (r,cl)   $
        \item Each resource is released exactly once: \label{bullet:rel-all}\\
              $  \forall_{r \in \resourseset} \: \exists_{n_1 \in \mathit{Nod}} \: \mapfunc(n_1) = (r,rl) \wedge    \nexists_{n_2 \in \mathit{Nod}} \: n_1\neq{}n_2 \wedge \mapfunc(n_2) = (r,rl)   $
        \item Every action node is preceded by a claim node on the corresponding resource:\\
              $  \forall_{n_1 \in \mathit{Nod}, a \in \mathbb{A}, p \in \mathbb{P}, r \in \resourseset} \: \mapfunc(n_1) = (a,p) \wedge \mathcal{R}(p) = r \Rightarrow  \exists_{n_2 \in \mathit{Nod}} \: \mapfunc(n_2) = (r,cl) \wedge   n_2 \rightarrow^+ n_1 $
        \item Every action node is succeeded by a release node on the corresponding resource:\\
              $  \forall_{n_1 \in \mathit{Nod}, a \in \mathbb{A}, p \in \mathbb{P}, r \in \resourseset} \: \mapfunc(n_1) = (a,p) \wedge \mathcal{R}(p) = r \Rightarrow  \exists_{n_2 \in \mathit{Nod}} \: \mapfunc(n_2) = (r,rl) \wedge   n_1 \rightarrow^+ n_2 $
        \item Every release node is preceded by a claim node on the corresponding resource:\\
              $  \forall_{n_2 \in \mathit{Nod}, r \in \resourseset} \:  \mapfunc(n_2) = (r,rl)  \Rightarrow  \exists_{n_1 \in \mathit{Nod}} \: \mapfunc(n_1) = (r,cl) \wedge n_1 \rightarrow^+ n_2 $
        \item Every claim node is succeeded by a release node on the corresponding resource:\\
              $  \forall_{n_1 \in \mathit{Nod}, r \in \resourseset} \:  \mapfunc(n_1) = (r,cl)  \Rightarrow  \exists_{n_2 \in \mathit{Nod}} \: \mapfunc(n_2) = (r,rl) \wedge n_1 \rightarrow^+ n_2 $
        \item No node precedes a claim node:\label{bullet:claim}\\
              $ \forall_{n_1 \in \mathit{Nod},   r \in \resourseset}   \:  \mapfunc(n_1) = (r,cl)  \Rightarrow   \nexists_{  n_2 \in \mathit{Nod} } \:   n_2 \rightarrow^+ n_1  $
        \item Release nodes are not succeeded by any node: \label{bullet:rel} \\
              $ \forall_{n_1 \in \mathit{Nod},   r \in \resourseset}  \:   \mapfunc(n_1) = (r,rl)  \Rightarrow   \nexists_{  n_2 \in \mathit{Nod} } \:   n_1 \rightarrow^+ n_2 $
        \item Every event node is preceded by another node: \label{bullet:event-precede}\\
              $ \forall_{n_1 \in \mathit{Nod}}  \:   \mapfunc(n_1) \in \eventset  \Rightarrow   \exists_{ n_2 \in \mathit{Nod}} \: n_2 \rightarrow^+ n_1  $
        \item Multiple instances of the same event in an activity must be sequentially ordered: \label{bullet:mult-events}\\
              $  \forall_{n_1,n_2 \in \mathit{Nod}} \: \mapfunc(n_1), \mapfunc(n_2) \in \eventset \wedge \mapfunc(n_1) = \mapfunc(n_2)  \Rightarrow  n_1 \rightarrow^+ n_2 \vee  n_2 \rightarrow^+ n_1$
    \end{enumerate}
    We denote the universe of activities with $\actsuni$.
\end{definition}

Nodes that are mapped to an action label are called action nodes.
Nodes that are mapped to a release or claim label are called resource nodes, and
nodes that are mapped to an event are called event nodes.

Constraints (\ref{bullet:claim}-\ref{bullet:mult-events}) extend the existing constraints of the activity framework \cite{sanden2016}.
Also constraints (\ref{bullet:cl-all}) and (\ref{bullet:rel-all}) are strengthened.
We explain why  (\ref{bullet:cl-all}), (\ref{bullet:rel-all}), (\ref{bullet:claim}), and (\ref{bullet:rel}) are needed when we discuss sequencing of activities.
Constraint (\ref{bullet:event-precede}) requires that event nodes are preceded by at least one node so that they depend
on at least one resource. This is because the event has to originate from
some component of the system and the framework uses the concept of resources to model system components.
Constraint (\ref{bullet:mult-events}) requires that multiple instances of the same event in an activity are
sequentially ordered. This ordering is necessary to maintain a deterministic behavior when we process the event instances.
If the instances are not sequentially ordered, the order of processing their outcomes is not well-defined and may lead to non-determinism.

We use an I/O automaton to specify the logistics controller.
An I/O automaton is defined as follows.
\begin{definition} \label{def:ioa}
    (I/O Automaton).
    An I/O automaton \cite{lynch1988introduction} is a seven-tuple ($Q$, $S$, $\Sigma$, $I$, $O$, $\Omega$, $F$) where
    \begin{itemize}
        \item  $\mathit{Q}$ is a finite set of states;
        \item  $\mathit{S} \subseteq \mathit{Q}$ is a set of initial states;
        \item  $\Sigma$ is a finite set of labels and $I$ and $O$ are disjoint sets of input and output labels respectively,
              such that $\mathit{I} \cup O = \Sigma$;
        \item  $\Omega \subseteq \mathit{Q} \times \mathit{I} \times O \times \mathit{Q}$ is a transition relation;
        \item  $\mathit{F}  \subseteq \mathit{Q}$ is a set of final states.
    \end{itemize}
    We write $q_1  \rightarrowtail q_2$ if $(q_1,i,o,q_2) \in \Omega$  for some $i$ and $o$, and denote the transitive closure of $ \rightarrowtail$ by $ \rightarrowtail^+$.
\end{definition}

With the ingredients of the activity framework defined, we formally define an activity specification as follows.

\begin{definition} \label{def:activity-specification}
    (Activity Specification).
    An activity specification is a five-tuple ($\acts $, $\eventset$, $\outcomeset $, $\gamma$, $\ioa$) where
    \begin{itemize}
        \item  $\acts$ is a finite set of activities with events from $\eventset$;
        \item  $\eventset$ is a finite set of events;
        \item  $\outcomeset$ is a finite set of event outcomes;
        \item  $\gamma \subseteq \eventset \times \outcomeset$ specifies the relation between events and their possible outcomes;
        \item  $\ioa$ is an I/O automaton with pairs of events and their outcomes as input labels and activities as output labels, i.e.,
              $I = \eventset \times \outcomeset \cup \{\lambda\}$ and $O = \acts \cup \{\epsilon\} $ where $\lambda$ is the empty input label denoting
              a transition that has no event and outcome pair, and $\epsilon$ is the empty activity $(\emptyset,\emptyset,\emptyset,\emptyset)$.
    \end{itemize}
\end{definition}

The max-plus semantics of an activity model with event feedback is defined in \cite{mohamadkhani}.

The running example system is modeled as follows.

\begin{itemize}
    \item $\mathbb{A} = \{ a,b,c,d,f,g \}$
    \item $\mathbb{P} = \{ p_1,p_2,p_3,p_4,p_5,p_6  \}$
    \item $\resourseset = \{ r_1,r_2,r_3 \}$
    \item $\acts = \{\mathit{Act}_1,\mathit{Act}_2,\mathit{Act}_3,\mathit{Act}_4  \}$ as given in Figure \ref{fig:activities-1234}
    \item $\eventset = \{e\}$
    \item $\outcomeset = \{u_1,u_2\}$
    \item $ \gamma = \{(e,u_1),(e,u_2)\}$
    \item I/O automaton $\ioa$ depicted in Figure \ref{fig:example-ioa} 
\end{itemize}

Note that the delay of $e$ is specified
in the DAG of activity  $\mathit{Act}_2$. Figure \nolinebreak\ref{fig:activities-1234} shows that
event $e$ occurs one time unit after the completion of $n_3$.

We assume the automaton does not process an event that
has not been emitted yet by an activity. We further assume that
the automaton has processed all events that are emitted before reaching a final state. This ensures that
the number of events emitted in the system cannot accumulate unlimited.
An I/O automaton is \emph{consistent} (Definition 18 of \cite{mohamadkhani}) if it follows these two assumptions
and does not contain states that do not lead to final states.
 
We further assume the logistics controller to be deterministic and complete:

\begin{definition} \label{def:DIOA}
    (Deterministic I/O Automata).
    A deterministic I/O automaton is an I/O Automaton ($Q$, $S$, $\Sigma$, $I$, $O$, $\Omega$, $F$) with
    \begin{itemize}
        \item exactly one initial state
              ($S=\{s\}$, $s\in{}Q$);
        \item no outgoing transitions from the final states;
              $\nexists_{(q_1,i,o,q_2) \in  \Omega} \:  q_1 \in F   $;
        \item and branches that are only based on different outcomes of the same event;
              $\forall_{d_1, d_2 \in \Omega, q_1,q_2,q_3 \in Q, i_1,i_2 \in I, o_1,o_2 \in O} \:  d_1 = (q_1,i_1,o_1,q_2), \:\:\: d_2 = (q_1,i_2,o_2,q_3),\:\: d_1 \neq d_2  \Rightarrow   \exists_{e \in \eventset, u_1,u_2 \in \outcomeset } \: i_1=(e,u_1), i_2=(e,u_2) \wedge (e,u_1), (e,u_2) \in \gamma  \wedge u_1 \neq u_2$.
    \end{itemize}
\end{definition}

Note that the restriction on branches makes sure that we do not have multiple outgoing transitions of any state with the same label.
Furthermore, it prohibits multiple outgoing transitions from any state with empty input labels ($\lambda$).
A complete I/O automaton is defined as follows.

\begin{definition} \label{def:CompIOA}
    (Complete I/O Automata).
    An I/O automaton is called \emph{complete} if at every branching state in the automaton in which an event is processed,
    there exists exactly one outgoing transition for each outcome of the processed event.
\end{definition}

\begin{figure}[t]
    \centering
    \scalebox{1}{
        \begin{tikzpicture}[shorten >=1pt,node distance=2cm,on grid,every node/.style={scale=0.9}]
            \node[state,initial,initial text=] (q_0) {$q_0$};
            \node[state] (q_1) [above right=1.6cm and 1.4cm of q_0] {$q_1$};
            \node[state] (q_2) [right=2.8cm of q_0] {$q_2$};
            \node[state,accepting] (q_3) [right=2.8cm of q_2] {$q_3$};
            \path[->]
            (q_0) edge [ left] node {$(\lambda,\mathit{Act}_1)$} (q_1)
            (q_1) edge [ right] node {$(\lambda,\mathit{Act}_2)$} (q_2)
            (q_2) edge [ above ] node {$((e,u_2), \mathit{Act}_4)$} (q_3)
            (q_2) edge [ below] node {$((e,u_1), \mathit{Act}_3)$} (q_0)
            ;
        \end{tikzpicture}
    }
    \caption{I/O automaton of the running example ($\ioa$)}
    \label{fig:example-ioa}
\end{figure}
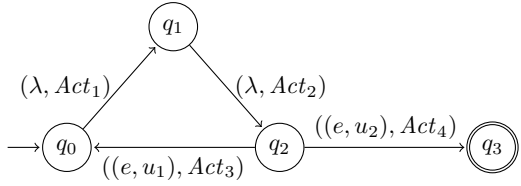

Figure \ref{fig:example-ioa} shows the I/O automaton of our running example.
The transitions are labeled as ((event, event outcome), activity).
The supervisor starts from state $q_0$ and first executes  $ \mathit{Act}_1,\mathit{Act}_2$. Then, reaching state $q_2$,
based on the outcome of event $e$, which is emitted by activity $\mathit{Act}_2$, it decides to
execute either $ \mathit{Act}_4$ in response to $u_2$ and goes to final state $q_3$, or $\mathit{Act}_3$ in response to $u_1$ and returns to $q_0$.
State $q_0$ is the initial state and $q_3$ is the only final state.
Note that nodes of activities that execute after the resolution of event $e$, cannot start until the event outcome is known. Hence, there exists an implicit dependency between those nodes and the node emitting $e$.
This is captured in the max-plus semantics of the model in \cite{mohamadkhani},
and our execution engine takes it into account.

\subsection{System Behavior}

Intuitively, the behavior of a system is defined by
an accepted word in the I/O automaton. It gives a sequence of activities that
the system executes in response to a sequence of event outcomes that are received from the plant.
An activity specification with event feedback
(Definition \ref{def:activity-specification}) and a deterministic I/O automaton (Definition \ref{def:DIOA})
consists of deterministic parts, which are sequences of activities that emit events,
and non-deterministic parts, which are the event outcomes
that are received from the plant and determine the activities that follow.
When an event outcome is received, the sequence of activities corresponding to the outcome is
given by a path in the automaton that starts with a transition that captures that
event outcome as input action.
Therefore, the complete behavior the system exhibits in a particular execution depends on the event outcomes that are received from the plant.

To define system behavior, we first define a word in an I/O automaton as follows.

\begin{definition} \label{def:word-in-ioa}
    (Words and Language of I/O Automata).
    A \emph{word} in an I/O automaton  $\ioa$ is a finite sequence of pairs $(i,o) \in I\times O$.
    The word $w = (i_1,o_1), (i_2,o_2),\dots,(i_j,o_j) \in (I\times O)^* $ is \emph{accepted} by $\ioa$
    if there exists
    a sequence $q_0, q_1, \dots, q_y \in Q^*$ of states such that $q_0 \in S$
    and $(q_{y}, i_{y+1},o_{y+1},q_{y+1}) \in \Omega_{\ioa}$, for all $0\leq y < j$,
    and $q_j \in F$. The language $\mathcal{L}(\ioa)$ is the set of all words accepted by $\ioa$.
\end{definition}

An accepted word in the automaton gives a sequence of activities.
The activity framework provides an operator, called the \emph{sequencing operator} (Definition \ref{def:sequence-op}), to combine
a sequence of activities into a single activity, giving the same order and timing of
action and event nodes. We use this operator to obtain the activity of a behavior (later defined in Definition \nolinebreak\ref{def:spec-behavior}).

Figure \ref{fig:activity-124} shows the results of the sequencing operator
sequencing $\mathit{Act}_{1}$, $\mathit{Act}_{2}$, and then $\mathit{Act}_{4}$, after
processing event $e$ with outcome $u_2$. The final result, activity $\mathit{Act}_{1.2.4}$, corresponds to the accepted word
$(\lambda,\mathit{Act}_1)$,$(\lambda,\mathit{Act}_2)$,$((e,u_2), \mathit{Act}_4)$.
Intuitively, the operator removes the release nodes of the left-hand side activity
and the claim nodes of the right-hand side activity, and adds dependencies from nodes that
precede the removed release nodes to nodes that succeed the corresponding claim node based on
the resource that they claim or release.
These claim and release nodes that are removed are colored red in $\mathit{Act}_{1.2}$ and $\mathit{Act}_{4}$.

Recalling from Definition \ref{def:activity}, we have two constraints to make
sure that (\ref{bullet:claim}) claim nodes have no dependency on any other node and (\ref{bullet:rel}) no nodes are dependent on release nodes.
Looking at Figure \ref{fig:activity-124}, we observe that
when sequencing two activities $\mathit{Act}_{1.2}$ and $\mathit{Act}_4$, dependencies
between nodes from $\mathit{Act}_{1.2}$ that violate (\ref{bullet:rel}), and
dependencies between nodes from $\mathit{Act}_4$ that violate (\ref{bullet:claim}) do not exist in  $\mathit{Act}_{1.2.4}$. The required dependencies are captured by new dependency relations between predecessors of the release nodes of the first activity and successors of the claim nodes of the second.

Note the dependencies from $n_5$ to $n_7$, and  from $n_5$ to the release node of $r_3$ in $\mathit{Act}_{1.2.4}$, drawn in red color.
These are added to capture the causality relation due to the outcome of $e$, produced by node $n_5$ and the corresponding sequence of activities that are executed in response (in this case, $\mathit{Act}_4$).
In general, when an event is processed, a dependency is added from the node that emits that event to every
succeeding node of all the claim nodes of the next activity to capture this causality.
This means that all nodes of the sequence that occurs in response to the outcome have to wait until the
outcome is resolved before they can execute.
The operator also enforces Constraint (\ref{bullet:mult-events}) for every event. Therefore,
it adds dependencies from every node that emits an event in the left-hand side activity
to every node that emits that same event in the right-hand side activity. This is not visible in the example, because there only a single event is emitted.
This is facilitated by strengthening Constraints  (\ref{bullet:cl-all}) and (\ref{bullet:rel-all}), compared to \cite{sanden2016}.

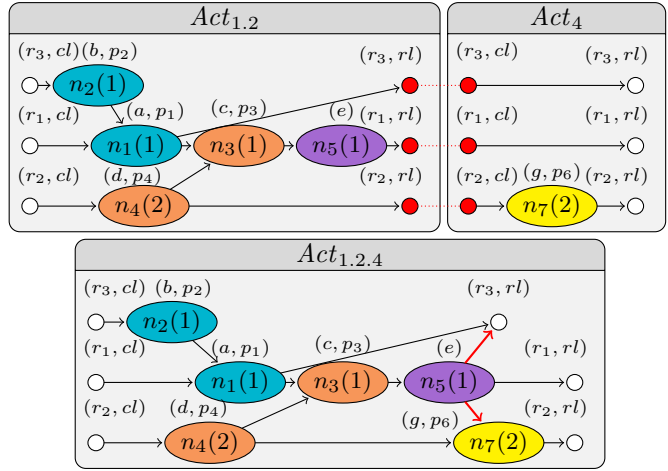
\begin{figure}[t]
    \centering
    \begin{tikzpicture}
        \tikzstyle{every state}=[  inner sep=1pt, minimum size=0pt]
        \node (n1) [activity]
        {
            \textbf{$\mathit{Act}_{1.2}$}
            \nodepart[align=left]{two}
            \begin{tikzpicture}[shorten >=1pt,on grid,auto]

                \node[state] (q_0) [fill=white, minimum size=6pt] {};
                \node[action rectangle] (l_0) [ above right=0.45cm and 0.25cm of q_0]{$(r_3,cl)$};
                \node[state] (q_1) [right=0.9 of q_0, fill=turquoise, style={font=\small}, style={draw,ellipse}] {$n_2 (1)$};
                \node[action rectangle] (l_0) [ above right=0.45cm and 0.15cm  of q_1]{$(b,p_2)$};
                \node[state] (q_2) [right=4.1 of q_1, fill=red, minimum size=6pt] {\tikzmark{r3rl}};
                \node[action rectangle] (l_0) [ above left=0.4cm and 0.25cm of q_2]{$(r_3,rl)$};

                \node[state] (q_3) [below=0.8 of q_0,fill=white, minimum size=6pt] {};
                \node[action rectangle] (l_0) [ above right=0.4cm and 0.25cm of q_3]{$(r_1,cl)$};
                \node[state] (q_4) [right=1.4 of q_3, fill=turquoise, style={font=\small}, style={draw,ellipse}] {$n_1 (1)$};
                \node[action rectangle] (l_0) [ above right=0.45cm and 0.25cm of q_4]{$(a,p_1)$};
                \node[state] (q_5) [right=1.35 of q_4, fill=peach , style={font=\small}, style={draw,ellipse}] {$n_3 (1)$};
                \node[action rectangle] (l_0) [ above=0.5 of q_5]{$(c,p_3)$};
                \node[state] (q_6) [right=1.35 of q_5, fill=MyLavender, style={font=\small}, style={draw,ellipse} ] {$n_5 (1)$};
                \node[action rectangle] (l_0) [ above=0.45 of q_6]{$(e)$};
                \node[state] (q_7) [right=0.9 of q_6, fill=red, minimum size=6pt] {\tikzmark{r1rl}};
                \node[action rectangle] (l_0) [ above left=0.4cm and 0.25cm of q_7]{$(r_1,rl)$};

                \node[state] (q_8) [below=0.8 of q_3,fill=white, minimum size=6pt] {};
                \node[action rectangle] (l_0) [ above right=0.4cm and 0.25cm of q_8]{$(r_2,cl)$};
                \node[state] (q_9) [right=1.5 of q_8, fill=peach, style={font=\small}, style={draw,ellipse} ] {$n_4 (2) $};
                \node[action rectangle] (l_0) [ above left=0.4cm and 0.15cm of q_9]{$(d,p_4)$};

                \node[state] (q_70) [below=0.8 of q_7, fill=red, minimum size=6pt] {\tikzmark{r2rl}};
                \node[action rectangle] (l_0) [ above left=0.4 and 0.25 of q_70]{$(r_2,rl)$};
                \path[->]
                (q_0) edge   (q_1)
                (q_4) edge   (q_2)
                (q_4) edge   (q_5)
                (q_3) edge   (q_4)
                (q_1) edge   (q_4)
                (q_5) edge  (q_6)
                (q_6) edge  (q_7)
                (q_8) edge  (q_9)
                (q_9) edge  (q_5)
                (q_9) edge  (q_70)
                ;
            \end{tikzpicture}
        };
        \node (n4) [activity, right=0.1 of n1]
        {
            \textbf{$\mathit{Act}_4$}
            \nodepart[align=left]{two}

            \begin{tikzpicture}[shorten >=1pt,on grid,auto]
                \node[state] (q_80) [fill=red, minimum size=6pt] {\tikzmark{r3cl}};
                \node[action rectangle] (l_0) [ above right=0.45cm and 0.25cm  of q_80]{$(r_3,cl)$};

                \node[state] (q_8) [below=0.8 of q_80,fill=red, minimum size=6pt] {\tikzmark{r1cl}};
                \node[action rectangle] (l_0) [ above right=0.4cm and 0.25cm  of q_8]{$(r_1,cl)$};

                \node[state] (q_4) [below=0.8 of q_8, fill=red, minimum size=6pt] {\tikzmark{r2cl}};
                \node[action rectangle] (l_0) [ above right=0.40cm and 0.25cm  of q_4]{$(r_2,cl)$};
                \node[state] (q_5) [right=1.1 of q_4, fill=yellow, style={font=\small}, style={draw,ellipse}] {$n_7(2)$};
                \node[action rectangle] (l_0) [ above=0.45 of q_5]{$(g,p_6)$};
                \node[state] (q_7) [right=1.1 of q_5, fill=white, minimum size=6pt] { };
                \node[action rectangle] (l_0) [ above left=0.4 and 0.25 of q_7]{$(r_2,rl)$};

                \node[state] (q_9) [above=0.8  of q_7, fill=white, minimum size=6pt] { };
                \node[action rectangle] (l_0) [ above left=0.4 and 0.25 of q_9]{$(r_1,rl)$};

                \node[state] (q_90) [above=0.8  of q_9, fill=white, minimum size=6pt] { };
                \node[action rectangle] (l_0) [ above left=0.4 and 0.25 of q_90]{$(r_3,rl)$};
                \path[->]
                (q_4) edge  (q_5)
                (q_5) edge  (q_7)
                (q_8) edge  (q_9)
                (q_80) edge  (q_90)
                ;
            \end{tikzpicture}
        };
        \node (n2) [activity , below right=0.1  and -4.8  of n1]
        {
            \textbf{$\mathit{Act}_{1.2.4}$}
            \nodepart[align=left]{two}
            \begin{tikzpicture}[shorten >=1pt,on grid,auto]

                \node[state] (q_0) [fill=white, minimum size=6pt] {};
                \node[action rectangle] (l_0) [ above right=0.45cm and 0.25cm of q_0]{$(r_3,cl)$};
                \node[state] (q_1) [right=1 of q_0, fill=turquoise, style={font=\small}, style={draw,ellipse}] {$n_2 (1)$};
                \node[action rectangle] (l_0) [ above right=0.45cm and 0.15cm  of q_1]{$(b,p_2)$};
                \node[state] (q_2) [right=4.3 of q_1, fill=white, minimum size=6pt] {};
                \node[action rectangle] (l_0) [ above=0.45cm of q_2]{$(r_3,rl)$};

                \node[state] (q_3) [below=0.8 of q_0,fill=white, minimum size=6pt] {};
                \node[action rectangle] (l_0) [ above right=0.45cm and 0.25cm of q_3]{$(r_1,cl)$};
                \node[state] (q_4) [right=1.9 of q_3, fill=turquoise, style={font=\small}, style={draw,ellipse}] {$n_1 (1)$};
                \node[action rectangle] (l_0) [ above=0.45 of q_4]{$(a,p_1)$};
                \node[state] (q_5) [right=1.35 of q_4, fill=peach , style={font=\small}, style={draw,ellipse}] {$n_3 (1)$};
                \node[action rectangle] (l_0) [ above=0.5 of q_5]{$(c,p_3)$};
                \node[state] (q_6) [right=1.4 of q_5, fill=MyLavender, style={font=\small}, style={draw,ellipse} ] {$n_5 (1)$};
                \node[action rectangle] (l_0) [ above=0.45 of q_6]{$(e)$};
                \node[state] (q_7) [right=1.65 of q_6, fill=white, minimum size=6pt] {};
                \node[action rectangle] (l_0) [ above left=0.45cm and 0.25cm of q_7]{$(r_1,rl)$};

                \node[state] (q_8) [below=0.8 of q_3,fill=white, minimum size=6pt] {};
                \node[action rectangle] (l_0) [ above right=0.45cm and 0.25cm of q_8]{$(r_2,cl)$};
                \node[state] (q_9) [right=1.5 of q_8, fill=peach, style={font=\small}, style={draw,ellipse} ] {$n_4 (2) $};
                \node[action rectangle] (l_0) [ above left=0.45cm and 0.15cm of q_9]{$(d,p_4)$};

                \node[state] (q_50) [right=3.8 of q_9, fill=yellow, style={font=\small}, style={draw,ellipse}] {$n_7(2)$};
                \node[action rectangle] (l_0) [ above left =0.30 and 0.9 of q_50]{$(g,p_6)$};
                \node[state] (q_70) [right=1 of q_50, fill=white, minimum size=6pt] { };
                \node[action rectangle] (l_0) [ above left=0.45 and 0.25 of q_70]{$(r_2,rl)$};
                \path[->]
                (q_0) edge   (q_1)
                (q_4) edge   (q_2)
                (q_4) edge   (q_5)
                (q_3) edge   (q_4)
                (q_1) edge   (q_4)
                (q_5) edge  (q_6)
                (q_6) edge  (q_7)
                (q_6) edge [thick, red] (q_50)
                (q_6) edge [thick, red] (q_2)
                (q_8) edge  (q_9)
                (q_9) edge  (q_50)
                (q_9) edge  (q_5)
                (q_50) edge   (q_70)
                ;
            \end{tikzpicture}
        };
    \end{tikzpicture}
    \begin{tikzpicture}[remember picture,overlay]
        \draw[red,densely dotted,->] (pic cs:r1rl) -- (pic cs:r1cl);
        \draw[red,densely dotted,->] (pic cs:r2rl) -- (pic cs:r2cl);
        \draw[red,densely dotted,->] (pic cs:r3rl) -- (pic cs:r3cl);
    \end{tikzpicture}
    \caption{Activities $\mathit{Act}_{1.2}$, $\mathit{Act}_{4}$ and $\mathit{Act}_{1.2.4}$}
    \label{fig:activity-124}
\end{figure}

In general, an activity may emit multiple instances of the same event. When we process an outcome of that event,
for reasons of determinacy, we need to take into account
which instance of the event is being processed.

Now that we have intuitively explained how the sequencing operator works, we define it formally.
To allow multiple instances of the same activity to be sequenced without confusing the nodes in their DAG, we use use superscripts to create unique names for the nodes of the composed activities.
If an event outcome is being processed, the operator also needs to capture the dependencies that arise between the event node of the left-hand activity that emits the event and
the nodes of the right-hand activity that are executed after processing that event.
We modify the sequencing operator of \cite{sanden2016} to be able to capture these dependencies as follows. To simplify the definition, we assume without loss of generality that all activities claim and release all resources by adding, if needed, a claim and release node for a resource without any other nodes in between.

\begin{definition} \label{def:sequence-op}
    Given two activities $A_1= ( \mathit{Nod}_1,\rightarrow_1, \mapfunc_1, \timefunc_1 )$ and
    $A_2= ( \mathit{Nod}_2,\rightarrow_2 , \mapfunc_2, \timefunc_2)$, an event
    $e$ being processed, and $k$ denoting the instance number of $e$ being processed.
    Let the sets of release nodes in $\mathit{Nod}_1$, and claim nodes in $\mathit{Nod}_2$ be
    $\mathit{rl}_{1} = \{ n_1 \in \mathit{Nod}_1  \mid  \exists_{r \in \resourseset} \: \mapfunc_1(n_1 ) =(r,rl)\}$ and
    $\mathit{cl}_{2} = \{ n_2 \in \mathit{Nod}_2  \mid  \exists_{r \in \resourseset} \: \mapfunc_2(n_2 ) =(r,cl)\}$,  respectively.
    We define $A_1 ;_{(e,k)} A_2$ as the activity $A_{1.2}=(\mathit{Nod}_{1.2},\rightarrow_{1.2}, \mapfunc_{1.2}, \timefunc_{1.2}) $, with:
    \begin{align*}
        \mathit{Nod}_{1.2} =   & \{n^1 \mid n \in \mathit{Nod}_1 \wedge n \notin  \mathit{rl}_{1} \} \cup                                            \\
                               & \{n^2 \mid n \in \mathit{Nod}_2 \wedge n \notin \mathit{cl}_{2}\}                                                   \\
        \rightarrow^S_{1.2} =  & \{ (n^1_i , n^1_j) \mid n_i \rightarrow_1 n_j \wedge n_j \notin  \mathit{rl}_{1}\} \cup                             \\
                               & \{ (n^2_i , n^2_j) \mid n_i \rightarrow_2 n_j \wedge n_i \notin  \mathit{cl}_{2}\} \cup                             \\
                               & \{ (n^1_1 , n^2_2) \mid (\exists_{n_{\mathit{rl}} \in \mathit{rl}_{1}} \: n_1 \rightarrow_1 n_{\mathit{rl}}) \wedge \\
                               & \;\;  (\exists_{n_{\mathit{cl}} \in \mathit{cl}_{2}} \: n_{\mathit{cl}} \rightarrow_2 n_2)  \} \cup                 \\
                               & \{ (n^1_1 , n^2_2) \mid \exists_{e' \in \eventset} \mapfunc_1(n_1) = e' \wedge  \mapfunc_2(n_2) = e'\}              \\
        \rightarrow^P_{1.2} =  & \{ (n^1_{e_k} , n^2_j) \mid n_{e_k} \in  \mathit{Nod}_1 \wedge \mapfunc_1(n_{e_k}) = e      \wedge                  \\
                               & \;\;\;   n_j \in  \mathit{Nod}_2 \wedge \exists_{n_{cl} \in \mathit{cl}_{2}} \: n_{cl} \rightarrow_2 n_j \wedge     \\
                               & \;\;\; |\{ n_e \mid \mapfunc_1(n_{e}) = e \wedge n_e \rightarrow_1^+ n_{e_k}\} |=k-1  \}                            \\
        \rightarrow_{1.2} \; = & \:\rightarrow^S_{1.2}  \cup \rightarrow^P_{1.2}                                                                     \\
        \mapfunc_{1.2} =       & \{ (n^1, \mapfunc_1(n)) \mid n \in \mathit{Nod}_1  \wedge n \notin  \mathit{rl}_{1}\} \cup                          \\
                               & \{ (n^2, \mapfunc_2(n)) \mid n \in \mathit{Nod}_2 \wedge n \notin  \mathit{cl}_{2} \}                               \\
        \timefunc_{1.2} =      & \{ (n^1, {\timefunc}_1(n)) \mid n \in \mathit{Nod}_1  \wedge n \notin  \mathit{rl}_{1}\} \cup                       \\
                               & \{ (n^2, {\timefunc}_2(n)) \mid n \in \mathit{Nod}_2 \wedge n \notin  \mathit{cl}_{2} \}                            \\
    \end{align*}
    If there are no events being processed when sequencing activities $A_1$ and $A_2$, then
    we denote the operation with $A_1 . A_2$, yielding activity  $A_{1.2}=(\mathit{Nod}_{1.2},\rightarrow^S_{1.2}, \mapfunc_{1.2}, \timefunc_{1.2}) $.
    This operator is well-defined for sequencing the empty activity, i.e., for any activity $A$,
    $A . \epsilon = \epsilon . A = A $.
\end{definition}

Note that in the final clause of the dependencies $\rightarrow^S_{1.2}$,
we add dependencies between all pairs of event-emitting nodes from both activities, although adding a single dependency from the last node
that emits the event in the left-hand side activity to the first one in the right-hand side activity would suffice, yielding an equivalent result, since the event emissions within both activities are already ordered. However, this would lead to a more complicated definition.

When processing an event, the sequencing operator uses
the event that is being processed and its instance number ($k$) to capture the causality dependencies due to the event processing, as explained earlier.

We define a function, called the \emph{behavior activity function},
that gives the activity of a behavior using the sequencing operator. 
We let $\emptysequence$ be the empty sequence (we use it for empty sequences of different types).

\begin{definition} \label{def:behavior-func}
    (Behavior Activity Function). Given an I/O automaton $\ioa$, 
    the behavior activity function $\behfunc: (I\times O)^* \rightarrow \actsuni$ is recursively defined as
    \begin{align*}
        \begin{cases}
            \behfunc(\emptysequence)=\epsilon           \\
            \behfunc(w(\lambda, o) )=     \behfunc(w).o \\
            \behfunc(w((e,u), o) )= \behfunc(w) ;_{(e,{|w|}^e+1)} o
        \end{cases}
    \end{align*}
    where ${|w|}^e$ denotes the number of times $e$ is processed in $w$.
\end{definition}

We use the behavior activity function and the sequencing operator in our execution to
construct the activity of a behavior (explained later in Definition \ref{def:spec-behavior}).
We describe how it is used for execution in detail in Section \ref{section:executionArch}.

An activity may emit multiple events with their outcomes resolved by the plant. However, the deterministic I/O automaton
only processes a single event at each branching state. This imposes an order in which these
event outcomes are processed which is given by the corresponding word in the automaton.
The sequence of (event, outcome) pairs of a behavior corresponds to the order in which
event outcomes are processed in the accepted word of the I/O automaton.
We provide a function that extracts this sequence from any given accepted word of the automaton as defined below.

\begin{definition} \label{def:proc-event-func}
    (Processed Events Function). Given an I/O automaton $\ioa$, 
    the processed events function $\Vee: (I\times O)^* \rightarrow  (\eventset \times \outcomeset)^*$ is recursively defined as
    \begin{equation*}
        \begin{cases}
            \Vee(\emptysequence)=\emptysequence \\
            \Vee(w(\lambda, o))=    \Vee(w)     \\
            \Vee(w ((e,u),o)) = \Vee(w)  (e,u)
        \end{cases}
        \\ \quad \\
    \end{equation*}
\end{definition}

\vspace{1em}

With the sequencing operator and the language of the automaton defined, we now define
the behavior of a system modeled in the activity framework with event feedback as follows.

\begin{definition} \label{def:spec-behavior}
    (Specified System Behavior).
    Given an activity specification ($\acts $, $\eventset$, $\outcomeset $, $\gamma$, $\ioa$),
    the specified system behavior $P$ is
    a word in $\mathcal{L}(\ioa)$. Such a word gives a sequence
    $\mathit{Act}_1,\mathit{Act}_2,\dots,\mathit{Act}_j$ of activities,
    and a corresponding sequence $U^P = (e_1,u_1),(e_2,u_2), \dots ,(e_i,u_i) \in (\eventset \times \outcomeset)^* $ of (event, outcome) pairs
    that lead to that sequence of activities.
    $\mathit{Act}^P = \behfunc(P)$ is the activity of behavior $P$ (Definition \ref{def:behavior-func}),
    and $U^P = \Vee(P)$ is the sequence of (event, outcome) pairs of $P$
    (Definition \ref{def:proc-event-func}).
\end{definition}

\subsection{System Timing}

With system behavior defined, now we focus on system timing.
Whenever we execute an activity, it takes time for its resources to execute the
individual actions. This means that executing an activity affects the
timing at which system resources are available.
In the activity framework, the state of the system constitutes the availability timing of each of the resources.
While the I/O automaton captures the possible sequences of activities, it does not
capture the timing of system resources at each state. This is determined by the activities. We capture the resource state by a function $\bar{x}$ 
that maps each resource to the availability time of that resource.
In the initial state $\bar{x}(0)$, all resources are assumed to be available at time $0$. This state is denoted as $\bar{0}$.

When we execute an activity, the resource availability function
is updated to represent the state of the system after executing the activity.
We define system timing as the timings of action and event nodes as follows.

\begin{definition}\label{def:system-timing}
    (System Timing).
    Given an activity ($\mathit{Nod}, \rightarrow \nolinebreak , \mapfunc, \timefunc$) and resource availability function $\bar{x}$,
    the specified start time $S(n,\bar{x})$ and
    completion time $C(n,\bar{x})$ of node $n \in \mathit{Nod}$ are defined by
    \begin{align*}
        S(n,\bar{x}) & =
        \begin{dcases}
            \bar{x} (r)                                   & \text{if }   n =(r,cl) \\
            {max}_{n'\in \mathit{Pred}(n)}  C(n',\bar{x}) & \text{otherwise}       \\
        \end{dcases} \\
        C(n,\bar{x}) & = S(n,\bar{x}) + \timefunc(n)
    \end{align*}
    Wherever $\bar{x}$ is known from context, we abbreviate $S(n,\bar{x}) , C(n,\bar{x})$ to  $S(n) , C(n)$.
    For $r\in \mathbb{R}$, we define the functions $S+r$ and $C+r$ as follows
    \begin{align*}
        (S+r)(n) & = S(n)+r \\
        (C+r)(n) & = C(n)+r
    \end{align*}
\end{definition}
\vspace{1em}

System timing for a behavior $P$ (Definition \ref{def:spec-behavior}), denoted by $S^P (\bar{x})$ and $C^P(\bar{x})$, can be computed
using the activity of a behavior and a given resource availability function.
Wherever $\bar{x}$ is known from context, we abbreviate to $S^P$ and $C^P$, respectively.


\section{Execution Relation to Preserve}\label{section:relation}

In this section, we define the relation that we aim to establish between the behavior and timing
prescribed by the specified model and the behavior and timing given by the execution.
We first define the timing relation.
Intuitively, we want our execution to adhere to the timing that the specified model gives.
This means the start times and the completion times of actions and events
in the model must be the same as their counterparts in execution, with some
error bound which we define in this section. We first define notations
for start and completion time of actions and events as well as their duration in the execution.

\begin{definition}\label{def:T-prime}
    (Execution Node Timings).
    Given an activity ($\mathit{Nod},\rightarrow, \mapfunc, \timefunc$), 
    for action nodes $\mapfunc(n)\in \mathbb{A} \times \mathbb{P}$, the
    \emph{execution start time}, $S'(n)$,
    is the wall clock time when the action starts on the real peripheral of the FMS,
    the \emph{execution completion time}, $C'(n)$,
    is the wall clock time when the action completes on the real peripheral of the FMS,
    and the \emph{execution duration} is $\timefunc'(n) = C'(n) - S'(n)$.
    For event nodes $\mapfunc(n)\in \eventset$,
    the \emph{event trigger time}, $S'(n)$, is the wall clock time at which the sampling of data starts,
    the \emph{event emission time}, $C'(n)$, is the wall clock time at which the sampling of data completes,
    and the \emph{event delay} is $\timefunc'(n) = C'(n) - S'(n)$.
\end{definition}

The activity framework defines start and completion time for actions, but there is no clearly defined relation
between model time and physical time. Because we deal with physical time, we need
to establish a relation between timing of actions and events in the model and their timing in physical execution.
We do so by first defining the start time of the execution as follows.

\begin{definition}\label{def:Psi}
    (Execution Start Time).
    The start time of the execution, denoted by $\Psi$, is the physical wall clock time, given by the user, at which the execution should start.
\end{definition}

Executing actions in reality has complications that the activity framework abstracts from. This includes
initializing the execution parameters and configuring the controller that is used to execute the model.
It also includes delays introduced by the laws of physics since it takes some time for action execution commands
to travel inside wires and reach the peripheral. Finally, the controller has to take some time to
compute the action and event nodes and find what to perform.

To address the physical delays in wires and the controller computation time, we define the following and show how we deal
with them in Section \ref{section:executionArch}.

\begin{definition}\label{def:exec-delay}
    (Execution Delay).
    Given an activity ($\mathit{Nod},\rightarrow\nolinebreak, \mapfunc, \timefunc$) and assuming $\bar{x}(0) = \bar 0$, for all $n \in \mathit{Nod}$, such that
    $\mapfunc(n)\in \mathbb{A} \times \mathbb{P} \cup \eventset $, the
    \emph{execution delay} of $n$ is
    \begin{equation*}
        \delta(n) = S'(n)-S(n)-\Psi
    \end{equation*}
\end{definition}

In Section \ref{section:executionArch}, we propose a plant model and an execution engine, that together form an
execution architecture. Assuming some inevitable innate delays are bounded, this architecture provides an upper bound $\DA$ on these execution delays such that
\begin{equation}\label{def:exec-delay-bound}
    \forall_{n \in \mathit{Nod}, \mapfunc(n)\in \mathbb{A} \times \mathbb{P} \cup \eventset} \: \delta(n) \leq \DA
\end{equation}
and uses the model time as a lower bound such that
\begin{equation*}
    \forall_{n \in \mathit{Nod}, \mapfunc(n)\in \mathbb{A} \times \mathbb{P} \cup \eventset} \:  S'(n) \geq S(n)+\Psi
\end{equation*}

With these definitions, we define the relation that we want to hold between the timings of the
action and event nodes in the model and their corresponding timings in the execution.

\begin{definition}\label{def:tim-preservation}
    (Timing Relation).
    Given two pairs of start and completion functions $S_1,C_1$ and $S_2,C_2$ all defined over the same domain the set $\mathit{Nod}$ of nodes,
    we denote their execution timing relation as
    \begin{align*}
         & S_1,C_1 \sim S_2,C_2 \iff                                                                      \\
         & \forall_{n \in \mathit{Nod}}    S_1(n) \leq S_2(n) \leq S_1(n) + \DA\wedge  C_2(n) \leq C_1(n)
    \end{align*}
\end{definition}

Now we define the relation that we want to preserve between the system behavior (Definition \ref{def:spec-behavior})
and the corresponding behavior that the execution engine executes.

Recall that a behavior is an accepted word of the I/O automaton and the activity and the sequence of
(event, outcome) pairs that are derived from it.
A deterministic I/O automaton (Definition \nolinebreak\ref{def:DIOA}) is deterministic up to states where non-deterministic event outcomes
give the next deterministic path. The behavior that is executed
is one that corresponds to an accepted word of the I/O automaton and is given by the sequence of
event outcomes resolved at runtime. Note that in a deterministic I/O automaton,
every accepted word corresponds to a unique path in the automaton. 

To define behavior preservation we first have to note that
the specified behavior of a system (Definition \ref{def:spec-behavior}), when executed,
boils down to action and event nodes that are executed and can be observed in the plant,
the outcomes of the events that were given by the plant, and the order in which the outcomes were processed.
Resource nodes are not observable in the plant and are used as a model abstraction that enable constraints
on the model to ensure correctness of behavior and validity of execution.
Therefore, executed system behavior, and consequently, behavior preservation do not take resource nodes into account.

\begin{definition} \label{def:exec-behavior}
    (Executed System Behavior). 
    The executed system behavior is
    a set $\mathit{Nod}'$ of all action and event nodes
    that are executed and their start and completion times $S', C'$ (Definition \ref{def:T-prime}),  
    and a sequence $U'$ of processed (event, outcome) pairs that is received from the plant.
\end{definition}

Intuitively, preserving a behavior is executing exactly all the nodes of its activity while
preserving all the dependencies between all the action and event nodes
and respecting the event outcomes that are resolved at runtime and processing them in the order
that the behavior specifies.
Since the event outcomes are determined at execution time by the plant, the
behavior that the execution engine preserves is also determined at runtime. Because the
I/O automaton is complete (Definition \ref{def:CompIOA}), all possible outcomes
of every event that is emitted and resolved during execution lead to an accepted word, and therefore,
to a specified system behavior.
We formally define behavior preservation as follows.

\begin{definition}\label{def:behavior-preservation}
    (Behavior Preservation).
    Given an activity specification ($\acts $, $\eventset$, $\outcomeset $, $\gamma$, $\ioa$),
    and starting state $\bar{x}$, let $P$
    be a specified behavior with ($\mathit{Nod}_P,\rightarrow_P, \mapfunc_P, \timefunc_P$) its corresponding activity,
    and $U^P$ its sequence of (event, outcome) pairs.
    We say $P$ is preserved, iff all action and event nodes $n \in \mathit{Nod}_{P}$ are executed, 
    no other actions or events are executed,
    the same sequence of event outcomes are received from the plant and processed, and
    for every dependency $n_1 \rightarrow_{P} n_2$, we have $C'(n_1) \leq S'(n_2)$.
\end{definition}

The primary goal of the execution engine and the execution architecture that we present in
Section \ref{section:executionArch} is to preserve any behavior of an activity specification by
respecting the event outcomes that correspond to the behavior, executing exactly
all event and action nodes of the activities that it is composed of,
and preserving node dependencies by preserving the start and completion times of those nodes. In
Section \ref{section:executionArch}, we present the assumptions on the plant, and the assumptions on the execution engine that
together, form the assumptions that we need for the execution architecture.
In Section \ref{subsecion:proof}, we prove that it preserves the timing relation (Definition \ref{def:tim-preservation})
for any behavior, and preservation of node dependencies follows trivially from that. We also explain how
it respects event outcomes and preserves the set of action and event nodes that it executes.

\section{Execution Architecture}\label{section:executionArch}

In this section, we define our execution architecture composed of the supervisory controller and the plant (Figure \nolinebreak\ref{fig:architecture}).
We first describe what we assume as the plant, then we describe the execution engine. Afterward, we show our approach using the running example.
Finally, we discuss the boundedness of the delays associated with executing a model, and we prove how our architecture
preserves the relation between the specification and execution.

\subsection{Plant}
For our execution engine to work properly, we need two assumptions on the plant.

\begin{assumption}\label{assum:actionExecTimeBounds}(Execution Duration Bounds) The plant performs a set of actions with known bounds on their execution durations. \end{assumption}
\begin{assumption}\label{assum:actionDelayBounds} (Execution Delay Bounds) There is a known bound on the amount of time that the plant takes
    to initiate actions or to process event data, and observe node completions. We denote the \emph{Execution Delay Bound} with $\DA$. \end{assumption}

\subsection{Execution Engine Description and Formalization}

The execution engine is configured with an activity specification (Definition \ref{def:activity-specification}).
When configured, it constitutes the supervisory controller. It has three layers each corresponding to the
different layers in the model.

At the top of the execution engine,
we have the Logistics Controller (LC). This layer is configured with the logistics controller specification (the I/O automaton)
and is responsible for executing it. In the middle layer of the engine, we have the Activity Controller
(AC), which is configured with activity specifications and is responsible for the execution
of activities. At the bottom layer of the engine, we have the Action Controller (aC), which
is configured with action and event node specifications and executes them on the plant using a time-triggered approach.

Recall that the model consists of deterministic parts, which are sequences of activities that emit events,
and non-deterministic parts, which are the event outcomes
that are received at runtime from the plant and determine the activities that follow. The execution engine iteratively executes deterministic
parts and processes the non-deterministic event outcomes from the plant.

The deterministic I/O automaton (Definition \ref{def:DIOA}) captures the deterministic parts of
the model with paths through the automaton with all intermediate states, if the path has any, only having one outgoing transition.
These paths may only have repetition for the first and last states, and they
end with a state that is either a final state or a state with multiple outgoing transitions.
The states with multiple outgoing transitions capture the non-deterministic part of the behavior with branches that correspond to
the different outcomes of the same event. So the overall behavior is deterministic up to event outcomes.
The branching states are
where the controller decides on the transition to take based on the event outcome that has been received.
We call these states \emph{decision states}.

\begin{definition}\label{def:decision-state}
    (Decision State).
    Given a deterministic I/O automaton ($Q$, $S$, $\Sigma$, $I$, $O$, $\Omega$, $F$), a \emph{decision state} is
    a state in $Q$ with more than one outgoing transition.
\end{definition}

In the running example (Figure \ref{fig:example-ioa}), state $q_2$ is a decision state.
The sequences that we execute correspond to a \emph{path} in the logistics automaton.

\begin{definition}\label{def:path}
    (Path). Given an I/O automaton ($Q$, $S$, $\Sigma$, $I$, $O$, $\Omega$, $F$), with $q_k \in Q$, $i_l \in I$, and $o_l \in O$,
    for $0\leq k\leq n$ and $1\leq l\leq n$, a \emph{path} is a sequence of
    4-tuples ${\rho} = (q_0,i_1,o_1,q_1),(q_1,i_2,o_2,q_2),\dots,(q_{n-1},i_{n},o_n,q_n)$ such that
    $(q_{j-1},i_{j},o_{j}, q_{j}) \in \Omega $ for each $1 \leq j \leq n$.
\end{definition}

The execution engine divides any path in the automaton into its constituent
deterministic parts called \emph{decision paths}.

\begin{definition}\label{def:decisionPath}
    (Decision Path).
    Given a deterministic I/O automaton ($Q$, $S$, $\Sigma$, $I$, $O$, $\Omega$, $F$),
    a path $(q_0,i_1,o_1,q_1),(q_1,i_2,o_2,q_2),\dots,(q_{n-1},i_{n},o_n,q_n)$ is a \emph{decision path}
    if $q_0$ is either an initial state or a decision state, and $q_n$ is either a final state
    ($q \in F$) or a decision state, and $q_1,\dots,q_{n-1}$ are not decision states.
\end{definition}

The decision paths of the running example are
${\rho}_1 = (q_0,\lambda, \mathit{Act}_1,q_1), $ $(q_1,\lambda,\mathit{Act}_2,q_2)$,
${\rho}_2 = (q_2,(e, u_1 ), \mathit{Act}_3,q_0),$ $ (q_0,\lambda, \mathit{Act}_1,q_1)$, $(q_1,\lambda,\mathit{Act}_2,q_2)$, and
${\rho}_3 = $ $(q_2,(e, u_2 ),\mathit{Act}_4,q_3)$.

Before the execution of a model starts, we have an initialization phase where the
execution engine is configured with the model. It starts when the LC
reads the transition labels of the I/O automaton starting from the initial state until it reaches a decision state or
a final state and then forms a decision path ${\rho}$.
The LC then sends the sequence of activity labels corresponding to ${\rho}$ to the AC.
We assume that logistics automata do not have states without paths to a final state.
Such states are redundant in specifying behavior, and they complicate model-driven execution.

\begin{assumption}\label{assum:blockingStates}
    (No Blocking States).
    All non-final states in the logistics automaton have a path to a final state. 
\end{assumption}


When the AC receives the
sequence of activity labels from the LC, it sequences them (based on Definition \nolinebreak\ref{def:sequence-op}).
The AC computes the start and completion time of every node based on the availability of resources and the
completion time of the event node that they depend on (given by the behavior activity function, Definition \ref{def:behavior-func}).
Then, the AC sends the labels of the action nodes and the event nodes of the resulting activity and their specified start and completion times to the aC
to be executed in a time-triggered manner. The initialization phase ends when the aC has received and prepared its first set of nodes and is ready to execute them.
To address the delay of the initialization phase, we assume the following.

\begin{assumption}\label{assum:Psi}
    (Execution Start Time).
    $\Psi$, is a moment in time far enough in the future that allows the initialization to take place before it.
\end{assumption}

Another assumption of our approach is that the execution duration of every action or the delay of every event
is at most equal to its specified duration (for actions) and specified delay (for events) in the model (including the execution delay upper bound $\DA$).
This ensures that the time-triggered execution of the aC keeps the timing and order of actions and events prescribed by the specification.
Formally, we write this assumption as follows.

\begin{assumption} \label{assum:T-prime-n} (Conservative Node Durations)
    $\forall_{ \mathit{Act} \in \acts , n \in \mathit{Nod}_\mathit{Act}} \: \timefunc(n) \geq \timefunc'(n) + \DA$
\end{assumption}

After the aC receives the node labels and their specified start and completion times from the AC,
it sorts them based on their intended start time $S(n)+ \Psi$.
Then, it compares its clock
$\psi$ with the intended start time of each node it has received from the AC, and
executes each action or event node when its time arrives ($\psi = S(n) + \Psi$).
To execute event nodes, the aC translates data from the plant
to event outcomes and sends those to the AC.
The aC also observes the execution of action and event nodes on the plant to check whether they complete
within their model-specified time (Assumption \ref{assum:T-prime-n}).
If an action or event node (in violation of the assumptions) does not complete before its specified completion time plus $\Psi$, appropriate action can be taken,
such as warning the user or logging the violation.

When the AC receives an event outcome from the aC, it notifies the LC of the outcome of the event immediately.
Then, the LC traverses the I/O automaton to the next decision state and sends the
corresponding sequence of activity labels to the AC, as well as the event being processed
and its instance ($e$ and $k$ of the sequencing operator, Definition \ref{def:sequence-op}).
These steps repeat until we reach a final state at which point the execution terminates.

Considering the delays that impact the execution times,
we let $\DE$ (event processing time) denote the maximum time that it takes for the
execution engine to receive an event outcome, find the next decision path, sequence the corresponding
activities, and prepare to execute the action and event nodes of that activity. Hence, $\DE$ denotes the maximum time required to process events.

To provide enough time for event processing,
we assume that the specified delay of an event ($\timefunc(e)$)
is larger than the time that it takes to resolve ($\timefunc'(e)$, Definition \ref{def:T-prime})
and process ($\DE$) the outcome. If a specification does not satisfy this assumption,
the controller will not have enough time to process the outcome, leading to the execution being delayed.
Recalling $\DA$ from Equation \ref{def:exec-delay-bound}, we formalize this assumption as follows.

\begin{assumption}\label{assum:T_e} (Conservative Event Delays)
    $\forall_{\mathit{Act} \in \acts , n \in \mathit{Nod}_\mathit{Act},  \mapfunc(n)\in \eventset} \: \timefunc(n) \geq \timefunc'(n) + \DE+ \DA$
\end{assumption}

\subsection{Example Execution}

Consider the running example again.
During the initialization phase, the LC is configured with $\ioa$ (Figure \ref{fig:example-ioa})
and the AC is configured with the activities (Figure \ref{fig:activities-1234}).
The LC starts from the initial state $q_0$ and traverses the automaton until it reaches a decision state.
Since $q_1$ has only one outgoing transition and is not a final state, LC continues to $q_2$ where there are multiple outgoing transitions with
event outcomes. Therefore, $q_2$ is a decision state where the LC stops traversing the automaton and sends the
activity-label sequence corresponding to ${\rho}_1$, $\mathit{Act}_{1},\mathit{Act}_{2}$ to the AC.

At this point, the AC obtains $\mathit{Act}_{1.2}$ using the sequencing operator and sends
the labels of the action  and event nodes and their start and completion times to the aC.
The aC now has the set $\{ (n_4, \Psi, \Psi+2), (n_2, \Psi, \Psi+1), (n_1, \Psi+1, \Psi+2), (n_3, \Psi+2, \Psi+3), (n_5, \Psi+3, \Psi+4)\}$ of
action and event nodes to execute. The aC sorts these nodes based on their intended start time $S(n)+ \Psi$.
When the clock of the aC $\psi$ reaches $\Psi$, it executes $n_4$ and $n_2$.
When $\psi = \Psi+1$, it executes $n_1$, and When $\psi = \Psi+2$, it executes $n_3$.
When $\psi = \Psi+3$, it samples the sensor corresponding to the event $e$, translates the data to
an outcome of $e$, and sends it back to the AC.

When the aC notifies the AC of the outcome of event $e$ and the node that gave it, the AC notifies the LC of the received
outcome and the LC traverses the corresponding path in the automaton to obtain the next
activity-label sequence to send to the AC.
This is given by the outgoing transitions of $q_2$ that process event $e$ emitted by $\mathit{Act}_2$ with
two possible outcomes, $u_1$ and $u_2$.
If we assume that the outcome $u_1$ is received, the LC sends the
activity-label sequence corresponding to decision path
${\rho}_2$. This brings the execution to a repeating sequence of activities ($\mathit{Act}_3, \mathit{Act}_1, \mathit{Act}_2$) being executed that
lasts as long as the outcome of $e$ is $u_1$ every time it is emitted.
If  $u_2$ is received, the LC sends the activity sequence corresponding to decision path
${\rho}_3$. In this case, the activity $\mathit{Act}_{4}$ is executed after which the execution halts.
Note that the AC is made aware by the aC what node emitted the event $e$. Therefore, it has the necessary information
to correctly sequence the following activities with their dependencies to the event node (depicted in red in Figure \ref{fig:activity-124}).

Figure \ref{fig:execution-act123} shows the execution of the decision path ${\rho}_1$, receiving event outcome
$u_1$, and consequently executing ${\rho}_2$ (for brevity, ${\rho}_2$ is only partially depicted).
There are three rows corresponding to system resources $r_1,r_2,r_3$.
Solid blocks on the upper side of each row depict model durations of action and event nodes,
while hatched ones on the lower side represent execution durations.
Solid blocks on the lower side at the start and end of each hatched block represent
execution delays which are together bounded by $\DA$.
Based on Assumption \ref{assum:T_e}, the aC sends the outcome
of event $e$ at time $C'(e)$ which occurs before $\Psi +4- \DE$ at latest.
Therefore, the execution engine has enough time to send the nodes of ${\rho}_2$ to the aC to be executed.
Thus, we hide $\DE$ and prevent it from causing any delay in the execution of nodes.
Note that nodes $n_4$ of  $\mathit{Act}_{2}$ and $n_6$ of $\mathit{Act}_{3}$ execute after the completion time of event node $e$
because they cannot start before $e$ is resolved and therefore, there is an implicit dependency between them.

\begin{figure}[t!]
    \centering
    \includegraphics[clip,width=0.8\columnwidth]{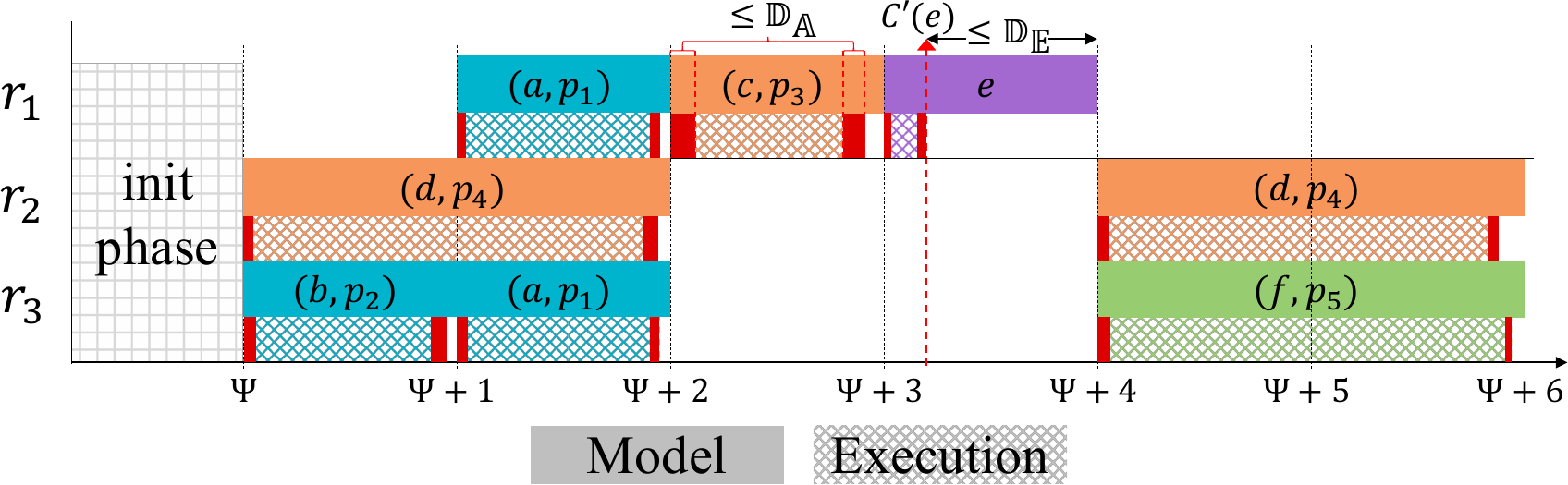}
    \caption{Partial execution Gantt chart of ${\rho}_1$,${\rho}_2$}
    \label{fig:execution-act123}
\end{figure}

\subsection{Delay Bounds}\label{subsection:delay-bounds}

Our architecture has delays that are bounded by $\DE$ and $\DA$.
Since the existence of bound $\DA$ is an assumed property of the plant, we determine the value of bound $\DE$.
We start by showing that our restrictions on the I/O automaton and activities ensure that the number of activities and
the number of nodes to be processed and executed in each iteration of the execution engine have a known bound.
Then, we show that the timing relation (Definition \ref{def:tim-preservation}) holds (Proposition \nolinebreak\ref{proposition:timing-rel}).

The I/O automaton is deterministic with finite sets of states, input labels and output labels.
If it does not contain any cycles, then the number of decision paths
that the LC reads is finite. If it does have cycles, then since a cycle cannot contain any final state,
and based on Assumption \ref{assum:blockingStates}, it must have a
branch that leads to a final state. Therefore, the state with
that branch is a decision state where the LC stops reading the
activity labels of that cycle and sends the decision path that it has obtained so far to the AC.
Therefore, cycles do not cause the specification to have an infinite number of decision paths.
This also makes the number of activity labels given by a decision path and sent to the AC finite
and bounded by the longest decision path which can be determined
from the specified I/O automaton.
Because an activity has a finite number of nodes, the number of nodes
to be executed at any time during the execution is finite.
This number is bounded by the largest number of nodes in any decision path.
Since when processing an event, the execution engine is dealing with a
bounded number of these bounded ingredients of the specification, then the time to
process an event must have a bound, which we call $\DE$.
$\DE$ may be computed using the following equation:
\begin{equation*}
    \DE = \Dee_\mathit{Event} + \DLC + \DAC + \DaC
\end{equation*}
where $\Dee_\mathit{Event}$ is the maximum time that it takes for an event outcome to be sent from
aC to the AC and finally to LC, $\DLC$ is the maximum time that the LC takes to read the largest decision path in the automaton
and find the activity label sequence and send it to the AC, $\DAC$ is the maximum time that the
AC takes to sequence the largest activity sequence and find node start and completion times and send them to the aC, and $\DaC$
is the maximum time that the aC takes to prepare the largest number of nodes in any decision path that it has received.
$\DLC, \DAC,$ and $\DaC$ depend on the size of the specification, while clearly,
all these values,  including $\Dee_\mathit{Event}$, depend on the speed of the processing platform and the implementation.

The LC performs a number of operations to read transition labels of a decision path and send them to the AC
that grows linearly with the number of transitions.
Let  $L(\rho)$ be the length (the number of transitions) of a decision path $\rho$ in the I/O automaton.
The complexity of the LC is ${\cal{}O}(L(\rho))$.
The AC performs the most complex operations of the entire engine.
Based on the activity labels received from the LC, it finds the DAG of each activity,
which depends on $L(\rho)$.
It then sequences all activities of $\rho$. The sequencing scales with the number of resources of the activities
we sequence, $\mathcal{R}(\mathit{Act})$, $\mathit{Act} \in \acts$, and the
number of nodes in them, see $\rightarrow_{1.2}$ in Definition \ref{def:sequence-op}.
The AC then sends node information to the aC, which is of complexity ${\cal{}O}( L(\rho) )$. Therefore,
$\DAC = {\cal{}O}(L(\rho) \times |\acts| + L(\rho) \times \mathcal{R}(\mathit{Act})^2 \times | \mathit{Nod}_{\mathit{Act}}|^2 )$.
Finally, for aC, we have $\DaC = {\cal{}O}( N(\rho) \log(N(\rho))  )$ where $N(\rho) $ is the number of nodes in $\rho$,
since it sorts the actions that it receives.

$\DE$ can be computed if the processing platform gives
hard real-time guarantees on its timings. For instance,
by measuring how long the LC takes to perform its operations for every transition, including reading and sending,
we can compute the time that it needs for every decision path. Using the longest decision path in the automaton,
we can compute the bound $\DLC$. If we use a processing platform without timing guarantees, $\DE$ may be
determined through measurements.


\subsection{Time- and Behavior-Preservation Proof}\label{subsecion:proof}

Recall that the aC implements a time-triggered execution. For time-triggered executions in general,
there are two cases to be made for execution start time
of every node. Either it has received the node and is prepared to execute it when its time arrives, or
it receives the node after its time arrives.
We denote the time at which the aC has all nodes of a decision path $\rho$ and is ready to execute them
by ${\timefunc^{aC}}(\rho)$ and formulate the above argument as follows.
\begin{equation} \label{eq:max-of-sn}
    \forall_{n \in \mathit{Nod}_{\rho}} \: S'(n) = \mathit{max} (S(n)+\Psi + \delta(n), {\timefunc^{aC}}({\rho}))
\end{equation}

\begin{figure}[t]
    \centering
    \includegraphics[clip,width=0.4\columnwidth]{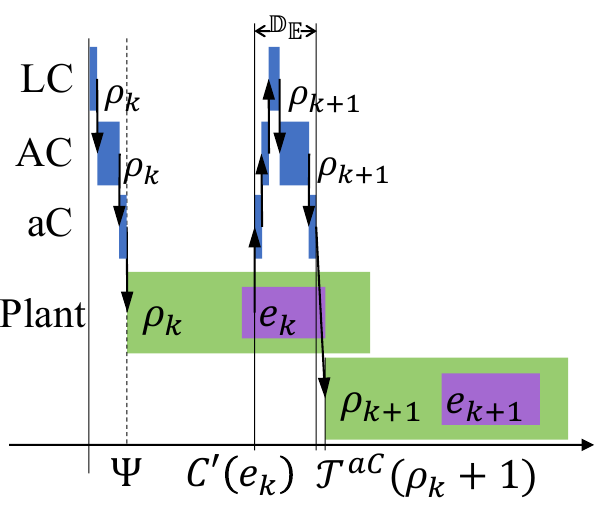}
    \caption{Execution chart of a sequence of decision paths}
    \label{fig:proposition}
\end{figure}

To formally prove that our execution engine satisfies the timing relation (\ref{def:tim-preservation})
between model and execution, we first need to introduce a few notations in relation
to the timing aspects of decision paths.
Since decision paths give sequences of activities and activities have timing
properties, we can derive the timing properties of a decision path
based on the timing of its activity sequence, given a starting state.
For that, we first introduce the start time of an activity.

\begin{definition}\label{def:start-time-of-activity}
    (Specified Start Time of an Activity).
    Given activity $A=(\mathit{Nod}_A,\rightarrow_A, \mapfunc_A, \timefunc_A)$ and
    resource availability function $\bar{x}$,
    the specified start time of activity $A$ is defined as
    \begin{equation*}
        S(A,\bar{x}) = \begin{cases}
            \mathit{min}_{i}\: \bar{x}_i,                           & \mathit{Nod}_{A} = \emptyset    \\
            \mathit{min}_{n \in \mathit{Nod}_{A} }\: S(n,\bar{x}) , & \mathit{Nod}_{A} \neq \emptyset
        \end{cases}
    \end{equation*}
    Wherever $\bar{x}$ is known from context, we abbreviate $S(A,\bar{x})$ to  $S(A) $.
\end{definition}

Based on the start time of the activity sequence of a decision path,
we define the start time of a decision path as follows.

\begin{definition}\label{def:start-time-of-dec-path}
    (Specified Start Time of a Decision Path).
    Given decision path
    ${\rho} = (q_0,i_1,\mathit{Act}_1,q_1),(q_1,i_2,\mathit{Act}_2,q_2),\dots,(q_{n-1},i_{n},\mathit{Act}_n,q_n)$ and
    resource availability function $\bar{x}$,
    the specified start time of decision path ${\rho}$ is defined as
    \begin{equation*}
        S({\rho},\bar{x}) =  S(\mathit{Act}_{1.2. \:\dots\: .n},\bar{x})
    \end{equation*}
    (with $\mathit{Act}_{1.2. \:\dots\: .n}$ defined in Definition \ref{def:sequence-op}).
    Wherever $\bar{x}$ is known from context, we abbreviate $S({\rho},\bar{x})$ to  $S({\rho}) $.
    We refer to start and completion times of action nodes in ${\rho}$ with
    $S^{\rho}$ and $C^{\rho}$, respectively.
\end{definition}

Given the notations above, now we prove that the execution engine satisfies the timing relation (Definition \ref{def:tim-preservation}).
We start by introducing a lemma that shows given our set of assumptions, the execution engine satisfies the
execution relation for any path in the I/O automaton. 
Then, since any specified behavior is a path made of a sequence of decision paths in the automaton,
we prove that the execution relation is kept by our execution engine for any activity specification with event feedback that satisfies our assumptions.
We use Figure \ref{fig:proposition} in our proof, which depicts the Gantt chart of a specification.

\begin{lemma} \label{lemma:induction}
    For any path ${\rho}={\rho}_1,{\rho}_2,\dots,{\rho}_i$ in an I/O automaton where ${\rho}_j$ are decision paths,
    for all $k \in \mathbb{N}, 1 \leq k \leq i$, we have $S^{{\rho}_{k}}+\Psi, C^{{\rho}_{k}}+\Psi \sim S'^{{\rho}_{k}}, C'^{{\rho}_{k}}$.
\end{lemma}

\begin{proof}
    We break down the relation into its constituent inequalities
    according to Definition \ref{def:tim-preservation} and prove them individually.
    For all $n \in \mathit{Nod}_{{\rho}_k}$ and assuming starting state $\bar{x}_0$, we prove
    \begin{align}
        S(n,\bar{x}_0)+\Psi & \leq S'(n,\bar{x}_0)  \label{eq:proof1}           \\
        S'(n,\bar{x}_0)     & \leq S(n,\bar{x}_0) + \Psi + \DA\label{eq:proof2} \\
        C'(n,\bar{x}_0)     & \leq C(n,\bar{x}_0) + \Psi \label{eq:proof3}
    \end{align}

    The aC uses a time-triggered execution mechanism that, by design, does not execute any action too early.
    Therefore, Equation \ref{eq:proof1} follows trivially from Equation \ref{eq:max-of-sn}.

    We prove Equations \ref{eq:proof2} and \ref{eq:proof3} by induction on the number of decision paths in ${\rho}$.
    For the base case $k=1$, the execution starts after all three layers
    have performed their tasks on ${\rho}_1$ and by definition, we have $\Psi = S'({\rho}_1)$.
    The aC has all nodes $n \in \mathit{Nod}_{{\rho}_1}$ that it needs to execute
    at time $\Psi$. This means that $\timefunc^{aC}({\rho}_1) \leq S(n) + \Psi$ for all $n \in \mathit{Nod}_{{\rho}_1}$.
    Based on the principles of time-triggered execution (Equation \ref{eq:max-of-sn}) and the upper bound on execution delays
    (Assumption \ref{assum:actionDelayBounds}), 
    we get  (for brevity, omitting $\bar{x}_0$ henceforth)
    \begin{equation*}
        S'(n) \leq S(n) + \Psi + \DA
    \end{equation*}
    Therefore, Equation \ref{eq:proof2} holds for $k=1$.
    Because node durations are conservative (Assumption \ref{assum:T-prime-n}) and
    they start execution within a given bound from their model time (Equation \ref{eq:proof2}),
    we add these two inequalities together to prove Equation \ref{eq:proof3}.

    \begin{equation*}
        S'(n) + \timefunc'(n)  + \cancel{\DA} \leq S(n) + \Psi + \cancel{\DA} + \timefunc(n)
    \end{equation*}
    By replacing the time durations $ \timefunc(n) $  and $ \timefunc'(n) $ by their start and completion times
    (Definition \ref{def:system-timing},\ref{def:T-prime}), we get
    \begin{equation*}
        \cancel{ S'(n)} + C'(n)  -\cancel{ S'(n)}    \leq \cancel{ S(n)}  + \Psi + C(n) -\cancel{ S(n)}
    \end{equation*}
    This means that Equation \ref{eq:proof3} also holds for $k=1$.

    Our induction hypothesis is that Equation \ref{eq:proof2} and Equation \ref{eq:proof3}
    hold for path $k$, and we prove the relation for the induction step to path $k+1$.
    Let $e_k$ be the event emitted by ${\rho}_k$, the outcome of which determines the next decision path.
    From the induction hypothesis, we know that $ S'(e_k) \leq S(e_k) + \Psi + \DA$.
    Since event delays are conservative, by adding Assumption \ref{assum:T_e} to this equation, we get
    \begin{equation*}
        S'(e_k) + \timefunc'(e_k) + \cancel{\DA} + \DE \leq S(e_k) + \Psi + \cancel{\DA} +\timefunc(e_k)
    \end{equation*}
    We replace the time durations $ \timefunc(n) $  and $ \timefunc'(n) $ by their start and completion times
    (Definition \ref{def:system-timing},\ref{def:T-prime}), and we get
    \begin{equation*}
        \cancel{ S'(e_k)} + C'(e_k) - \cancel{S'(e_k)}+  \DE \leq \cancel{S(e_k)} + C(e_k) - \cancel{S(e_k)}+ \Psi
    \end{equation*}
    Because $C(e_k) \leq S({\rho}_{k+1})$, and $C'(e_k) +  \DE=\timefunc^{aC}({\rho}_{k+1})$, we can write
    \begin{equation*}
        {\timefunc^{aC}}({\rho}_{k+1}) \leq S({\rho}_{k+1}) + \Psi
    \end{equation*}
    From Definition \ref{def:start-time-of-dec-path}, we conclude that
    \begin{equation*}
        \forall_{n \in \mathit{Nod}_\mathit{{\rho}_{k+1}} } \:  S'(n)  \leq S(n) + \Psi + \DA
    \end{equation*}
    This means that Equation \ref{eq:proof2} holds for $k+1$.
    The proof for Equation \ref{eq:proof3} holding for $k+1$ follows the same identities as for the base case.
    We conclude that Equations \ref{eq:proof2} and \ref{eq:proof3} hold for case $k+1$.
    This concludes the induction proof, and consequently, the complete proof.
\end{proof}

\begin{proposition} \label{proposition:timing-rel}
    For any activity specification with event feedback (Definition \ref{def:activity-specification}),
    under Assumptions \ref{assum:actionExecTimeBounds}-\ref{assum:T_e},
    let $\Psi$ be the start time of the execution. If $S^P, C^P$ are specified start and completion times,
    and $S'^P, C'^P$ are the execution start and completion times
    of the nodes in a behavior $P$ (Definition \ref{def:spec-behavior}) of the specification, then $S^P+\Psi, C^P+\Psi \sim S'^P, C'^P$.
\end{proposition}

We prove Proposition \ref{proposition:timing-rel} using Lemma \ref{lemma:induction}.

\begin{proof}
    To prove Proposition \ref{proposition:timing-rel}, we observe that any behavior $P$ of an activity specification with event feedback
    is derived from an accepted sequence of decision paths in the automaton.
    Since Lemma \ref{lemma:induction} shows that any sequence of decision paths is executed on time,
    then $P$ is executed on time. Therefore, $S^P+\Psi, C^P+\Psi \sim S'^P, C'^P$.
\end{proof}

Note that preserving dependencies between
nodes follows as a corollary of Proposition \ref{proposition:timing-rel}.

\begin{corollary}\label{corollary:behavior-preservation}
    For any behavior $P$ of a given activity specification with event feedback,
    if Proposition \ref{proposition:timing-rel} holds for the start and completion times of action and event nodes in $P$,
    then for all action and event nodes $n_1, n_2 \in \mathit{Nod}_{P}$ such that $n_1 \rightarrow n_2$, we have $C'(n_1) \leq S'(n_2)$, meaning that
    all node dependencies in $P$ are preserved.  \QEDB
\end{corollary}

By design, the execution engine
respects the values of the event outcomes that correspond to the behavior that it executes
and processes them in the order specified by the behavior, and
it executes all the action and event nodes of that behavior.
Therefore, the execution engine preserves the behavior that it executes
(Definition \ref{def:behavior-preservation}).
We then formulate the main contribution of this article in a theorem that covers time and behavior preservation
by our execution architecture as follows.

\begin{theorem}\label{theorem:final}
    For any activity specification with event feedback (Definition \ref{def:activity-specification}),
    under Assumptions \ref{assum:actionExecTimeBounds}-\ref{assum:T_e}, 
    any behavior that is executed, 
    is executed in a time- (Definition \ref{def:tim-preservation}) and
    behavior-preserving (Definition \ref{def:behavior-preservation}) manner.   \QEDB
\end{theorem}

\section{Validation}\label{section:validation}

In this section, we validate our approach on the xCPS platform (Figure \ref{fig:xcps}) by demonstrating
that the execution engine executes the xCPS SC specification in a time- and behavior-preserving manner.
The SC specification is an activity model (Definition \nolinebreak\ref{def:activity-specification}) of the xCPS.

Figure \ref{fig:xcps-gantt} shows a Gantt chart of a behavior of the xCPS
producing two products (for brevity), where horizontal bars
correspond to actions on peripherals and the arrows in between them are the dependencies.
The figure shows a non-trivial behavior and the challenges that arise due to
\emph{synchronization},  \emph{concurrency},  \emph{timing} and \emph{pipelining}, which we
discussed in Section \ref{section:motivexample}.
There, we briefly explained the details of how a product is assembled.
Configured with the specification of the xCPS SC, the execution engine becomes the supervisory controller,
performing all necessary actions to produce products, while
preserving time and behavior for all action and event nodes.

\begin{figure}[t!]
    \centering
    \scalebox{1}{
        \includegraphics[clip,width=0.99\columnwidth]{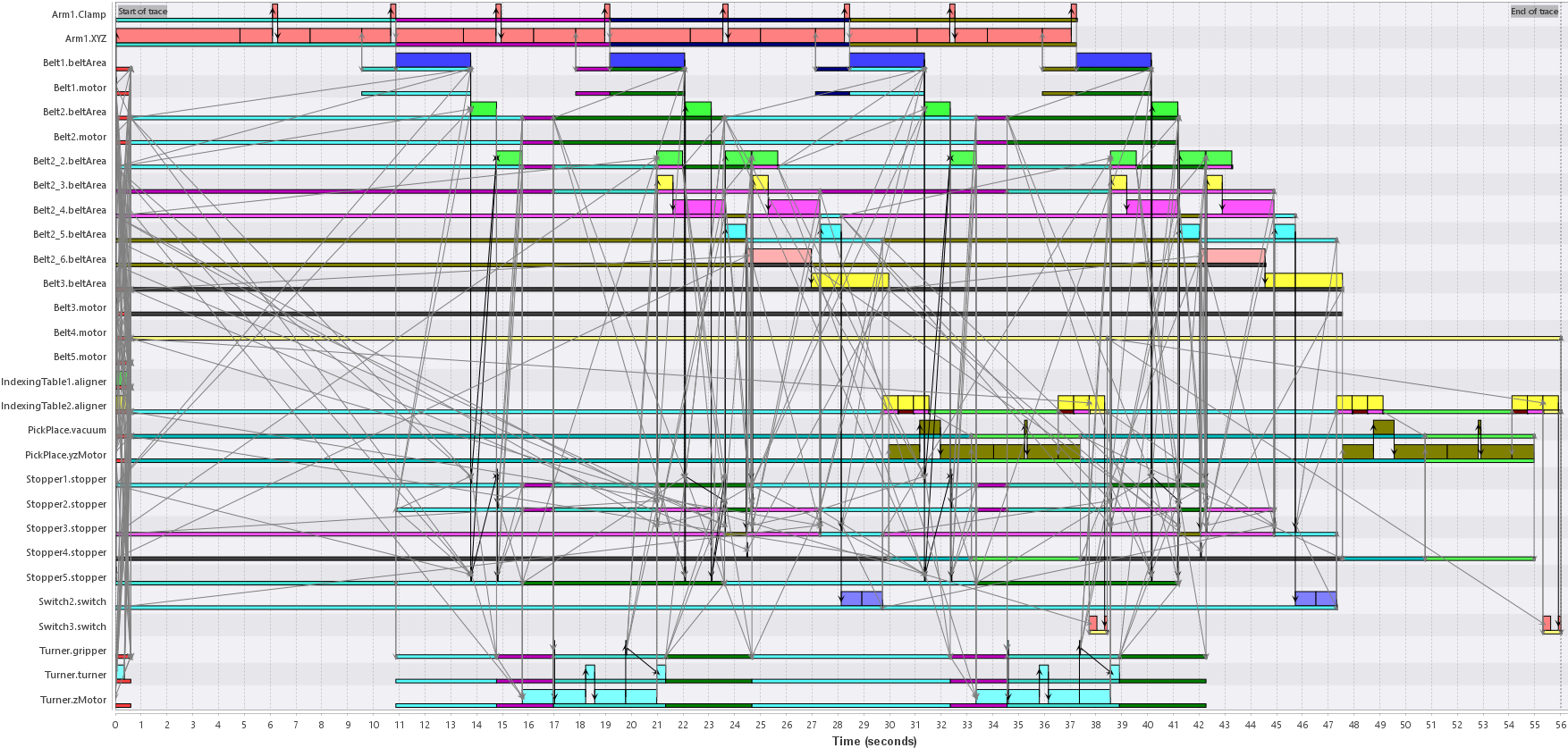}
    }
    \caption{Gantt chart of the xCPS producing two products}
    \label{fig:xcps-gantt}
\end{figure}

To properly control the xCPS we need to preserve the model start and completion times with timing errors
bounded by a time value, denoted by $\mathbb{B}$, determined by time-critical actions of the xCPS.
As explained in Section \ref{section:motivexample}, the act of pushing pieces to the indexing table using switch 2
is the most time-sensitive action of the xCPS. This action has a time window of 20 milliseconds to be correctly
executed, computed based on belt speed and distance. Missing this time window means
the piece can no longer be pushed to the table (Figure \ref{fig:xcps-schematic}).
Therefore, we take $\mathbb{B}=20$ milliseconds. The execution engine must satisfy this property to ensure
proper control of the machine.

Recall that behavior is preserved since, through timing preservation,
a node never starts earlier than the
completion time of a node it depends on, all nodes of a behavior are executed, and all event outcomes of a
behavior are respected (conform Theorem \ref{theorem:final}).

Given the specification of the xCPS, we must check whether it satisfies Assumptions \ref{assum:T-prime-n} and \ref{assum:T_e}.
Therefore, we first establish the execution delay bound ($\DA$)
of the nodes, as well as the event processing delay bound ($\DE$).
Then, we determine the worst-case execution duration ($T'(n)$) of actions and the delay of events in the system
(Definition \ref{def:T-prime}).
With these values, we can determine if the specified duration of action nodes and the specified delay of event nodes,
satisfy Assumptions \ref{assum:T-prime-n} and \ref{assum:T_e}, respectively.

To establish the execution delay bound ($\DA$), we measure the
response time of the aC for every node and check how long the aC takes to process
each action or event node.
Establishing $\DA$ is also necessary to check if we can meet the time-criticality of
nodes in the xCPS such that $\DA \leq \mathbb{B}$.
Then, from these measurements,
we determine $\DA$. There are variations in these measurements because we run
our engine on a Linux machine that has tasks to perform beside the
execution engine, and they take time. 
We observe that the aC takes consistently less than $1.4$ milliseconds to process nodes.
To allow some margin, we determine $\DA = 1.6$ milliseconds as the bound that
the execution engine guarantees on the execution delay of nodes.
Note that establishing $\DA$ also satisfies Assumption \nolinebreak\ref{assum:actionDelayBounds}.

Next, we establish $\DE$ based on measurements on the execution engine.
Since we run the engine on a Linux machine that cannot guarantee timing, we cannot compute $\DE$
based on the size of the specification. We use the largest
decision path in the specification based on the total number of nodes in
its activities to establish the bound $\DE$.
We discussed in Section \ref{subsection:delay-bounds}
why that is the largest contributor to the time that it takes to process an event.
By giving the largest decision path to the engine in a test scenario (not connected to the physical
machine), and measuring the time that it takes to sequence all the activities in that decision path,
we can determine individual components that make up
$\DE$. Note that the values of $T(n)$ have no effect on the results of this test.
The largest decision path of the xCPS specification has 75 nodes.
We observe $\Dee_\mathit{Event} < 0.5$, $\DLC < 1.8$, $\DAC < 2.8$, and $\DaC < 0.8$ milliseconds
for the processing time of this decision path. Therefore, we determine $\DE = 6$ milliseconds (rounded upwards).
Further, we observe the time needed to resolve the assembler event is $T'(e) <1.9$ milliseconds, which we then use to check the
event delays of the model to satisfy Assumption \ref{assum:T_e}.
The specified delay of the event node of the assembler
is $10$ milliseconds which satisfies Assumption \ref{assum:T_e} since $10>1.9+6+1.6 $.

To determine $T'(n)$ for action nodes,
we examined the timings of the actions that the
xCPS performs. Some actions, such as movements of the gantry arms, are controlled by a dedicated precision motion controller
that moves the arms very accurately. These actions do not need timing measurements
since we can determine the bound on their timing based on the configuration of their controller.
For actions that have very little variation in their timing, we performed 20 measurements
and  for the ones that have large variation in their timing, we performed 100 measurements.
In both cases, for every action, the bound on their timing was determined to be the maximum over all measurements
of that action (Assumption \ref{assum:actionExecTimeBounds}).

Having the $T'(n)$ of all action nodes, $\DA$, and $\DE$, we then check
their timing in the model.
The xCPS specification satisfies Assumption \ref{assum:T-prime-n} for all nodes.

Our execution engine guarantees time and behavior preservation under Assumptions \ref{assum:actionExecTimeBounds}-\ref{assum:T_e}.
So far, we have determined $\DA$, $\DE$, and $T'(n)$ of all action and event nodes for the xCPS,
satisfying Assumptions \ref{assum:actionExecTimeBounds}, \ref{assum:actionDelayBounds}, \ref{assum:T-prime-n}, and  \ref{assum:T_e}.
Now we check the rest of our assumptions, and then, we execute the complete specification.

To ensure that the I/O automaton of the xCPS has no blocking states (Assumption \ref{assum:blockingStates}),
it was synthesized using the approach of \cite{sanden2015}
with modular automata that describe the order of activities executed in each station of the xCPS.
For example, an automaton describes the different allowed orders of activities that can be executed
around the indexing table to safely align it with the belts for pushing products on and off of it, as well
as aligning it with the assembler to assemble a product. The approach is based on
SC synthesis theory of \cite{ramadge1987supervisory}, and ensures non-blockingness
because transitions leading to blocking states are pruned during synthesis. We synthesized an I/O automaton
to produce three products that has 1,169 states and 3,221 transitions. A deterministic I/O automaton was then
derived from the result and was given as logistics automaton to the execution engine.
This was done by pruning transitions that violated the restrictions of a deterministic I/O automaton.
For every state, we only allowed one outgoing transition unless they were based on different outcomes of the same event.
Then we removed the states and transitions that became unreachable because of the pruning.
As a result of the synthesis, the I/O automaton had only one initial and one final state and no outgoing transition
from the final state. The deterministic I/O automaton has 59 states and 61 transitions.

To ensure that the engine has sufficient time to initialize (Assumption \ref{assum:Psi}),
we measure the initialization time of the execution engine by initializing the complete specification of the xCPS.
Since our implementation is based on a Linux machine, we perform this test 20 times.
It consistently takes less than 800 milliseconds to prepare every layer
and start executing the nodes. Therefore, we choose $\Psi$ to be the wall clock time one second
in the future from when we start the execution.

\begin{figure}[t!]
    \centering
    \scalebox{0.9}{
        \includegraphics[clip,width=0.6\columnwidth]{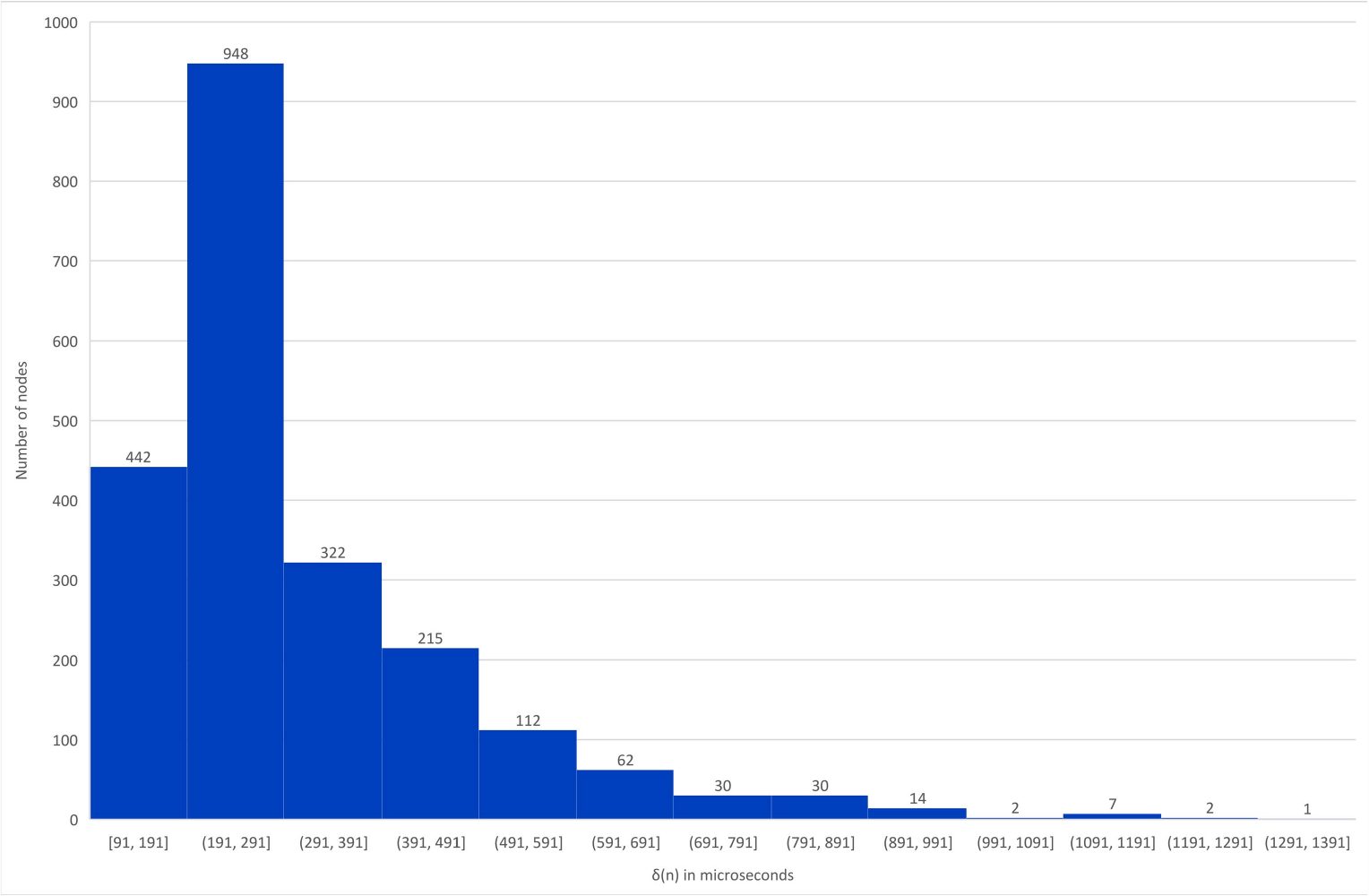}
    }
    \caption{$\delta(n)$ for all nodes over 10 runs of the xCPS}
    \label{fig:xcps-deltas}
\end{figure}

Now that all the specification and the plant assumptions are satisfied, we run the complete specification 10 times, each time producing three products.
Figure \ref{fig:xcps-deltas} shows the distribution of $\delta(n)$ for all action and event nodes over 10 runs of the system.
Every run executes 241 action and event nodes, 2410 in total of 10 runs.
We observe that $\DA \leq \mathbb{B}$. Meaning that the execution engine guarantees an
execution delay bound for all nodes that satisfies the time-criticality of actions in the xCPS.
We further observe that a great majority (71\%) of action and event nodes are executed less than 391 microseconds after their model start time.
We notice that $\delta(n)$ increases when there are many nodes starting at the same time and the hardware concurrency
of our controller does not provide enough resources to handle all of them concurrently.

In conclusion, this shows that we have delivered the promise of Theorem \ref{theorem:final} on a physical FMS.
All action and event nodes, including the most time-sensitive ones, are executed within their specified start and completion times,
meaning that Proposition~\ref{proposition:timing-rel} is confirmed in a physical setting.
The execution engine executes the complete specification of the xCPS. It executes all nodes of a behavior, respects the
outcomes that correspond to the behavior, and preserves node timings and dependencies.

\section{Related Work}\label{section:relatedWrok}
In this section, we discuss the literature relevant to our work. 
We view our work as an automated way to synthesize a model-driven implementation given a set of assumptions on the
plant and the model. There exist two main types of approaches to implementation synthesis.
One set of approaches aim to
configure an engine that is designed specifically for a given model
of computation (\cite{ptides, van2015scenario, zhou1992design}),
and other approaches rely on code generation (\cite{lohstroh2021toward, amnell2004times}).
Our work fits within the former category.
Implementation synthesis approaches may preserve a variety of properties, such as behavior and timing.
Our work preserves timing as well as behavior, where behavior is preserved through timing guarantees and respecting event outcomes.
In the end, we compare the functionality of the aC with the time-triggered architecture \nolinebreak\cite{TTA}.

The logistics controller of our execution engine is configured with a deterministic I/O automaton.
We do not make any assumptions on how the logistics automaton is obtained.
The I/O automaton could be synthesized using the SC theory of \nolinebreak\cite{ramadge1987supervisory}.
In such an approach, the activities would be considered \emph{controllable}
actions that may emit events, and the possible outcomes of those events would be considered \emph{uncontrollable} actions
which we respond to by a sequence of activities.
Such an approach would guarantee desirable properties such as non-blockingness.

Similar to our approach where the specification is used to configure an execution engine,
an actor-based model is used to configure an execution engine for distributed real-time embedded systems in \cite{ptides}.
Every actor has a set of input ports and output ports and a function that captures the timing relation between them. 
The framework provides a mechanism to decide whether it is safe to process an input event at any input port or not, in
a distributed execution setting. 
Every event has a deterministic logical timestamp, decoupled from physical timing variations. \emph{Real-time ports} provide a guaranteed relation between
the real time at which they process an event and the logical timestamp of that event. This method ensures a deterministic execution, independent from non-deterministic physical (distributed) timing variations.
Activities can be mapped to the actor-based model of \nolinebreak\cite{ptides}.
However, such a mapping would lose the scalability and analysis advantages of our model that
come from the multiple levels of abstraction for specification (actions, activities, automata).

Another example of configuring an execution engine with a specification can be found in \cite{van2015scenario},
where Finite State Machine-based Scenario-Aware Data-Flow (FSM-SADF) \cite{fsmsadf}
is used to specify a multi-processor system.
The execution architecture is configured by allocating resources to actors.
Since the approach uses time-division multiplexing (TDM), these resources include
processor cores and memory, as well as time slots of those cores to guarantee timing.  
The number of scenarios to be executed affects resource budgets and switching time.
Pipelining of multiple scenarios is possible and scenarios do not overtake each other.
The execution of scenarios is deterministic and data-driven which allows a predictable timing
and worst-case guarantees for them. Our approach also provides worst-case behavior
timing guarantees assuming worst-case timing for actions.
Our execution engine follows a time-driven execution where model-prescribed system timing
of nodes and a given $\Psi$ determine execution timings. Similar to \cite{van2015scenario},
timing overhead of our engine also increases
with the size of the model. We hide this overhead in $\DE$.

Petri nets \cite{peterson1981petri} are a well-known formalism with potential
as a means to configure an execution engine.
An example can be found in \cite{zhou1992design} where they are directly executed on an FMS.
A Petri net is systematically specified such that properties such as liveliness and
boundedness are guaranteed. A supervisory controller is then
given from this Petri net that preserves these properties.
Similar to our architecture, this is executed on the controller
of an FMS that controls the individual actions of the system.
The controller runs a Petri net execution engine that is configured with the specification of the
supervisory controller. The order of actions is preserved by the execution engine while timing of actions is not considered.
We assume timed actions and guarantee timing and behavior preservation between the model and the
execution. Furthermore, we enable the supervisory controller to make decisions at runtime based on different outcomes of events,
which is not possible in \cite{zhou1992design}.

Now we move on to code generation methods such as \cite{amnell2004times}, where
timed automata \cite{alur1994theory} are used for code generation for real-time systems.
The model is extended with task graphs to describe dependencies between tasks. Resources of a task are mapped to semaphores
that are used to ensure proper resource use. The automaton is extended with guards that have clock (i.e., time) and variable constraints.
Similar to our work, a deterministic automaton is used. Determinism is ensured by enforcing priorities to action transitions
that are simultaneously enabled and could cause non-deterministic behavior in the automaton.
Similar timing assumptions to ours, namely, Assumptions \ref{assum:actionDelayBounds} and \ref{assum:actionExecTimeBounds}, are made.
The three levels of abstraction in our framework for specification (actions, activities, automata)
provide more scalability in specification and analysis (using max-plus linear matrix-vector multiplications)~\cite{sanden2016}.

Another code-generation method, which is considered a follow-up of \cite{ptides}, can be found in \cite{lohstroh2021toward},
which focuses on providing a deterministic execution for concurrent actors (called \emph{reactors} in the framework) that react to incoming events carrying logical time stamps.
Similar to our work, a DAG is used to specify dependencies between reactors.
The framework provides code generation for a number of different programming languages.
Every reactor is then mapped to a block of generated code with guarantees that
concurrent reactors are executed in a deterministic manner, meaning that a set of inputs
always gives a predictable set of outputs. Similar to behavior preservation in our approach,
perturbations in the timing of reactors do not change the behavior of the system.  
Similar to \cite{ptides}, activities can be mapped to the model described in \cite{lohstroh2021toward},
with the advantage of giving a deterministic execution as opposed to \nolinebreak\cite{ptides}.
However, the same discussion on scalability advantages (concerning specification and analysis)
of our framework over \nolinebreak\cite{ptides}, can be made for \cite{lohstroh2021toward}.

The aC uses time-triggered execution comparable to the
time-triggered architecture introduced in \cite{TTA}, where real time is modeled
with instances of time with fixed intervals in between them, called \emph{durations}, from the past to the future.
A duration can be of activity (when the modeled system is active)
or of silence (when it is not). 
In our approach, time is not divided into fixed intervals.
The time instances of our model are derived from the specification timings.
Our execution engine follows these timings and does not incorporate a predefined
set of fixed durations where it is active or silent.
The aC continuously observes its clock time and at the model-prescribed time
of a node plus $\Psi$, sends signals to
actuators that perform actions (for action nodes), and samples data from sensors (for event nodes).
In other words, instead of being triggered by the arrival of the next fixed instance of time,
as is the case in \cite{TTA}, the aC is triggered by the start and completion times of the nodes that it executes.

\section{Conclusion}\label{section:conclusion}

In this article, we turned the activity framework into a model-driven approach by introducing an execution architecture and 
execution engine that allow a specification to be executed in a time- and behavior-preserving fashion.
We proved that the architecture and engine preserve the specified ordering of action and event nodes in 
the specification as well as the timing of these nodes up to a specified bound, while respecting the sequence of 
event outcomes that are received from the plant. 
We demonstrated our approach on the xCPS, a prototype production system.

\section{Acknowledgement}\label{section:ack}
 
This research was funded by EU ECSEL Joint Undertaking under grant agreement no 826452 (project Arrowhead Tools).

\section{Conflicts of Interest} 

The authors have no conflict of interest to declare that are relevant to this article.

\bibliography{sn-bibliography}

\end{document}